\documentclass[11 pt]{article}

\usepackage[utf8]{inputenc}

\usepackage{fullpage}
\usepackage[colorlinks=true,linkcolor=black,citecolor=black,urlcolor=black]{hyperref}
\usepackage{amsmath,amsfonts,amssymb,amsthm}
\usepackage{latexsym}
\usepackage{lineno}

\usepackage{graphicx}
\usepackage{enumitem}
\usepackage{amsthm}
\usepackage{hyperref}
\usepackage{algorithm,algpseudocode}
\usepackage{mathtools}
\usepackage{float}
\usepackage{xcolor}
\usepackage{subcaption}
\graphicspath{ {figs/} }

\def\defn#1{\textit{\textbf{\boldmath #1}}}

\newtheorem{theorem}{Theorem}
\newtheorem{lemma}[theorem]{Lemma}
\newtheorem{corollary}[theorem]{Corollary}
\newtheorem{proposition}[theorem]{Proposition}

\newtheorem{observation}[theorem]{Observation}
\theoremstyle{definition}

\theoremstyle{remark}

\newtheorem*{note*}{Note}

\newtheorem*{remark*}{Remark}
\newtheorem{claim}[theorem]{Claim}
\newtheorem*{claim*}{Claim}

\newcommand{\remove}[1]{{}}

\newcommand{\anna}[1]{{\color{cyan} Anna: #1}}

\newcommand{\whereis}[1]{{}}

\title{Forbidding Edges between Points in the Plane to Disconnect the Triangulation Flip Graph
}

\usepackage{authblk}
\author[1]{Reza Bigdeli}
\author[1]{Anna Lubiw}
\affil[1]{Cheriton School of Computer Science, University of Waterloo, Canada. {\tt \{rbigdeli,alubiw\}@uwaterloo.ca}}

\begin{document}

\maketitle


\begin{abstract}

The \defn{flip graph} for a set $P$ of  points in the plane has a vertex for every triangulation of $P$, and an edge when two triangulations differ by one \defn{flip} that replaces one triangulation edge by another.  
The flip graph is known to have some connectivity properties:
(1) the flip graph is connected; 
(2) connectivity still holds 
when restricted to triangulations containing some \emph{constrained} edges between the points; 
(3) for $P$ in general position of size $n$, the flip graph is $\lceil \frac{n}{2} -2 \rceil$-connected, a recent result of Wagner and Welzl (SODA 2020).

We introduce the study of  
connectivity properties of the flip graph when some edges between points are \emph{forbidden}.  An edge $e$ between two points is a \defn{flip cut edge} if eliminating triangulations containing $e$ results in a disconnected flip graph.  More generally, a set $X$ of edges between points of $P$ is a \defn{flip cut set}
if eliminating all triangulations that contain edges of $X$ results in a disconnected flip graph. 
The \defn{flip cut number} of $P$ is the minimum size of a flip cut set.

We give a characterization of flip cut edges that leads to
an $O(n \log n)$ time
algorithm to test if an edge is a flip cut edge and, with that as preprocessing, an $O(n)$ time algorithm to test if two triangulations are in the same connected component of the flip graph. 
For a set of $n$ points in convex position (whose flip graph is the 1-skeleton of the associahedron) we prove that the flip cut number is $n-3$.
\end{abstract}

\section{Introduction}
Given a set $P$ of $n$ points in the plane, which may include collinear points, an \defn{edge} of $P$ is a line segment $pq$ 
that intersects $P$ in exactly the two endpoints $p$ and $q$.
A \defn{triangulation} of $P$ is maximal set of non-crossing edges. 
Triangulations 
have important applications in graphics and mesh generation~\cite{bern1995mesh,edelsbrunner2001geometry} and are of significant mathematical interest~\cite{de2010triangulations}.  

A fundamental 
approach to understanding
triangulations is 
by means of \emph{flips}.
A \defn{flip} operates on a triangulation by removing one edge $pq$ and adding another edge $uv$ to obtain a new triangulation---of necessity, the edges $pq$ and $uv$ will cross and their four endpoints will form a convex quadrilateral with no other points of $P$ inside it. 
For example, in Figure~\ref{fig:first-flip-cut-edge}, edge $a_1 b_1$ can be flipped to $uv$. 
In 1972, Lawson~\cite{lawson1972generation,LAWSON1972365}
proved that any triangulation of point set $P$ can be 
reconfigured to any other  triangulation of $P$ by a sequence of flips.
This can be expressed as connectivity of the \defn{flip graph}, which has a vertex for every triangulation of $P$ and an edge when two triangulations differ by a flip. 
Reconfiguring triangulations via flips is well studied~\cite{bose2009flips}, 
but there are 
some very interesting open questions, and many properties of flip graphs remain to be discovered.

There is considerable work on  distances and diameter of flip graphs~\cite{eppstein,hanke1996distance,hurtado1999flipping,kanj2017computing,pournin2014diameter, sleator1988rotation}.
The worst-case diameter of the flip graph is $\Theta(n^2)$ for general point sets, and
$2n-10$ for $n > 12$
points in convex position~\cite{pournin2014diameter,sleator1988rotation}.
For a general point set, finding the distance in the flip graph between two triangulations (the ``flip distance'') is known to be NP-hard~\cite{LUBIW201517}, 
and even APX-hard~\cite{Pilz2014flip}.  However, 
a main open question is 
the complexity of the flip distance problem for convex point sets (is it NP-hard or in P?).

The case of points in convex position is especially interesting because there is a bijection between flips in triangulations of a convex point set and rotations in binary trees~\cite{sleator1988rotation}, so that flip distance becomes rotation distance between binary trees.
Finding the rotation distance between two binary trees is of great interest in biology for phylogenetic trees~\cite{dasgupta1997distances}, and in data structures for splay trees~\cite{sleator1988rotation}.
Furthermore, the flip graph
for $n$ points in convex position is the 1-skeleton of an ($n-3$)-dimensional polytope called the \emph{associahedron}~\cite{LEE1989551}, or see~\cite{ceballos2015many}. See Figure~\ref{fig:A_5}.
Although there is no geometric analogue of the associahedron for the case of triangulations of a general point set,  
some of its properties carry over to an abstract complex called the \emph{flip complex}.  For example, the 2-dimensional faces of the flip complex, like those of the associahedron, have size 4 or 5~\cite{lubiw2019orbit}.

An open frontier in the study of flip graphs has to do with expander properties, which would potentially lead to rapid mixing via random flips. 
For results on mixing in triangulations, see~\cite{caputo2015,molloy2001mixing, randall2000analyzing}.
It has recently been shown that the reconfiguration graph of bases of a matroid is an expander~\cite{anari2018log}, and it would be interesting to know if similar results hold for triangulation flip graphs.
More generally, researchers study connectivity properties of flip graphs.
Recently, Wagner and Welzl~\cite{wagner2020connectivity} showed that for $n$ points in general position in the plane, the flip graph is $\lceil \frac{n}{2} -2 \rceil$-connected.
For points in convex position, the flip graph is $(n-3)$-connected, which follows from Balinski's theorem~\cite{balinski1961graph} applied to the 1-skeleton of the associahedron, see~\cite{wagner2020connectivity}.

One intriguing thing about flip graphs of triangulations is that many properties carry over when we restrict to triangulations containing some specified non-crossing edges---so-called \emph{constrained triangulations}.  The subgraph of the flip graph consisting of triangulations that contain all the constrained edges is 
connected, as proved by Chew~\cite{chew1989constrained} using the \emph{constrained Delaunay triangulation}.

\begin{figure}[bt]
    \centering
    \includegraphics[scale=1.0]{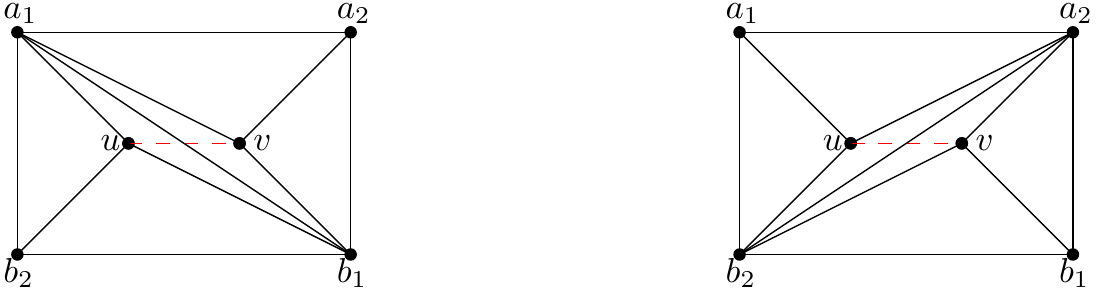}
    \caption{The smallest point set that 
    has a flip cut edge. The edge $e=uv$ is a flip cut edge since forbidding $e$ leaves two possible triangulations (as shown) and neither one allows a flip.
    }
    \label{fig:first-flip-cut-edge}
\end{figure}

\begin{figure}[bht]
    \centering
    \includegraphics[scale=0.6]{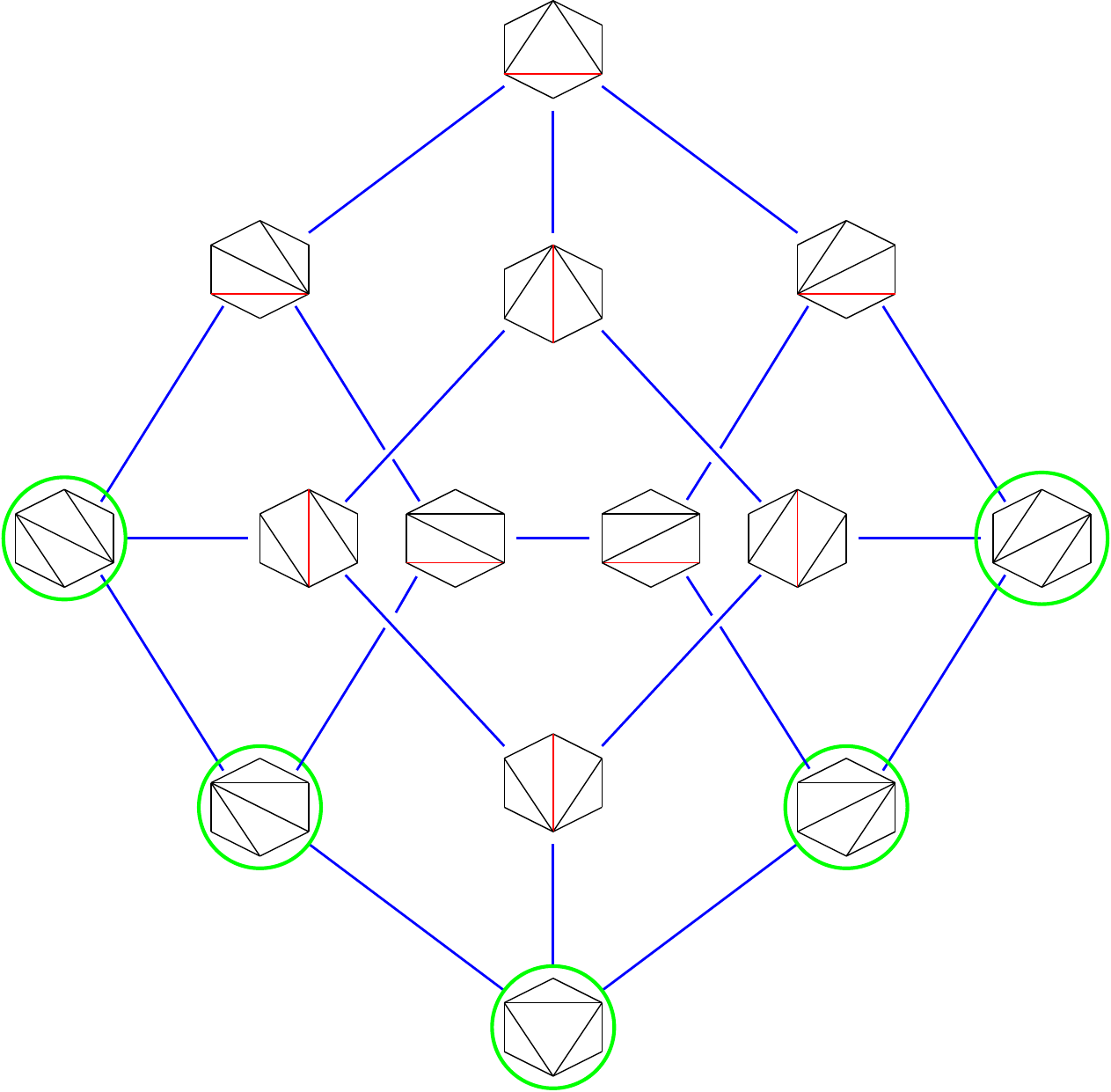}
    \caption{
    The flip graph of points of a convex hexagon  is the 1-skeleton of an associahedron.
    If we forbid the two red edges, the resulting flip graph (with vertices circled in green) is 
    connected. 
    }
    \label{fig:A_5}
\end{figure}

\medskip\noindent{\bf Our Results.}
In this paper we study connectivity properties of the flip graph when---instead of constraining certain edges between points to be present---we \emph{forbid} certain edges between points. To be precise, if a set $X$ of edges between points is forbidden, we eliminate all triangulations that contain an edge of $X$, and examine whether the flip graph on the remaining triangulations is connected.  
We say that $X$ is a \defn{flip cut set} if the resulting flip graph is disconnected; in the special case where $X$ is a single edge, we say that the edge is a \defn{flip cut edge}.  For example the edge $uv$ in Figure~\ref{fig:first-flip-cut-edge} is a flip cut edge, but the two red edges in Figure~\ref{fig:A_5} do not form a flip cut set.
We define the \defn{flip cut number} of a set of points to be the minimum size of a flip cut set.  This is analogous to the \emph{connectivity} of a graph---the minimum number of vertices whose removal disconnects the graph.

Since the structure of the flip graph depends on the edges between the points, it seems more natural to study connectivity of the flip graph after deleting some of these edges, rather than deleting some vertices of the flip graph, as standard graph connectivity does, and as the result of Wagner and Welzl~\cite{wagner2020connectivity} does.

As our main result, we characterize when an edge $e$ is a flip cut edge in terms of connectivity (in the usual graph sense) of the edges that cross $e$. Observe that a triangulation that does not contain $e$ must contain an edge that crosses $e$. 
We then use the characterization to give an $O(n \log n)$ time algorithm to test if a given edge $e$ in a point set of size $n$ is a flip cut edge.  With that algorithm as preprocessing, we give a linear time algorithm to test if two triangulations are still connected after we eliminate from the flip graph all triangulations containing edge $e$.

For the case of $n$ points in convex position, there are no flip cut edges and we show that  the flip cut number is $n-3$.  For example, in Figure~\ref{fig:A_5} the leftmost and rightmost triangulations become disconnected if we forbid one more edge, which yields a flip cut set of size $3$ for $n=6$.

\begin{figure}[th]
    \centering
    \includegraphics[scale=0.45]{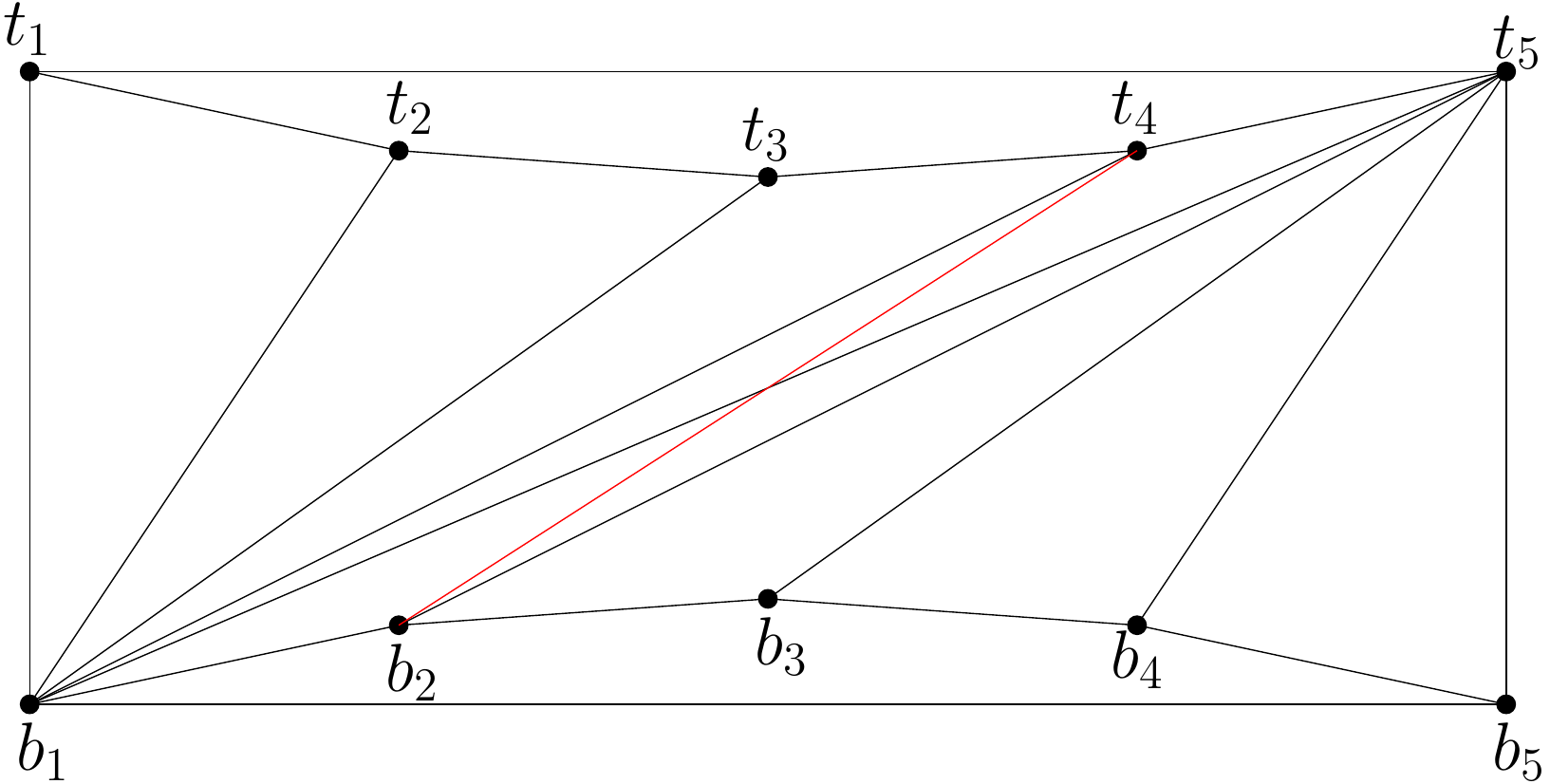}
    \caption{The ``channel'', and a triangulation that becomes \defn{frozen} (an isolated vertex in the flip graph) if we forbid the edge $b_2,t_{n-1}$ (in red). In fact, every edge $b_i t_j$, $i,j \notin \{1,5\}$ is a flip cut edge.
}
    \label{fig:channel}
\end{figure}

We show that a point set of size $n$ may have $\Theta(n^2)$ flip cut edges (see Figure~\ref{fig:channel}), and we show that a flip cut edge may result in $\Theta(n)$ disconnected components in the flip graph.  
We also examine various special point sets whose flip graphs have been previously studied,
such as points on an integer grid~\cite{caputo2015} and, more generally,  point sets without empty convex pentagons~\cite{eppstein}.
Our characterization of flip cut edges becomes simpler in the absence of empty convex pentagons.  Point sets without empty convex pentagons must
have collinear points; our results do not assume points in general position.

\subsection{Examples} 
\label{sec:examples}

Figure~\ref{fig:first-flip-cut-edge} shows the smallest example of a point set that has a flip cut edge---in fact, by forbidding one edge, we get a triangulation from which no flips are possible.  Such an isolated vertex in the flip graph is called a \defn{frozen} triangulation.

More generally, the triangulation of a \defn{channel} shown in Figure~\ref{fig:channel} becomes frozen (modulo triangulating the upper and lower convex subpolygons) when we forbid
one flip cut edge. 
In fact, any edge joining an interior vertex of the top curve and an interior vertex of the bottom curve is a flip cut edge.  So there are $\Theta(n^2)$ flip cut edges. We justify this more carefully in Section~\ref{sec:channels}.
Channels have been studied previously---they provide lower bounds on the maximum diameter of the flip graph~\cite{hurtado1999flipping}, and are building blocks for proving NP-hardness of flip distance~\cite{LUBIW201517,Pilz2014flip}.

Grid points are another well-studied case for triangulation flips~\cite{caputo2015}, in part because of physics applications.
For $n$ points lying on a $\sqrt{n} \times \sqrt{n}$ grid,
there are again $\Theta(n^2)$ flip cut edges. 
In fact, for an infinite grid, every edge is a flip cut edge, though boundary effects interfere in finite grids.
See
Section~\ref{sec:grids}. 
Points in a grid have the special property that there are no empty convex pentagons.  Flips for  point sets without empty convex pentagons were studied by Eppstein~\cite{eppstein}.  Such point sets must have collinear points~\cite{abel2011every}. 
See Section~\ref{sec:no-EC5}.

\begin{figure}[t]
    \centering
    \includegraphics[scale=0.6]{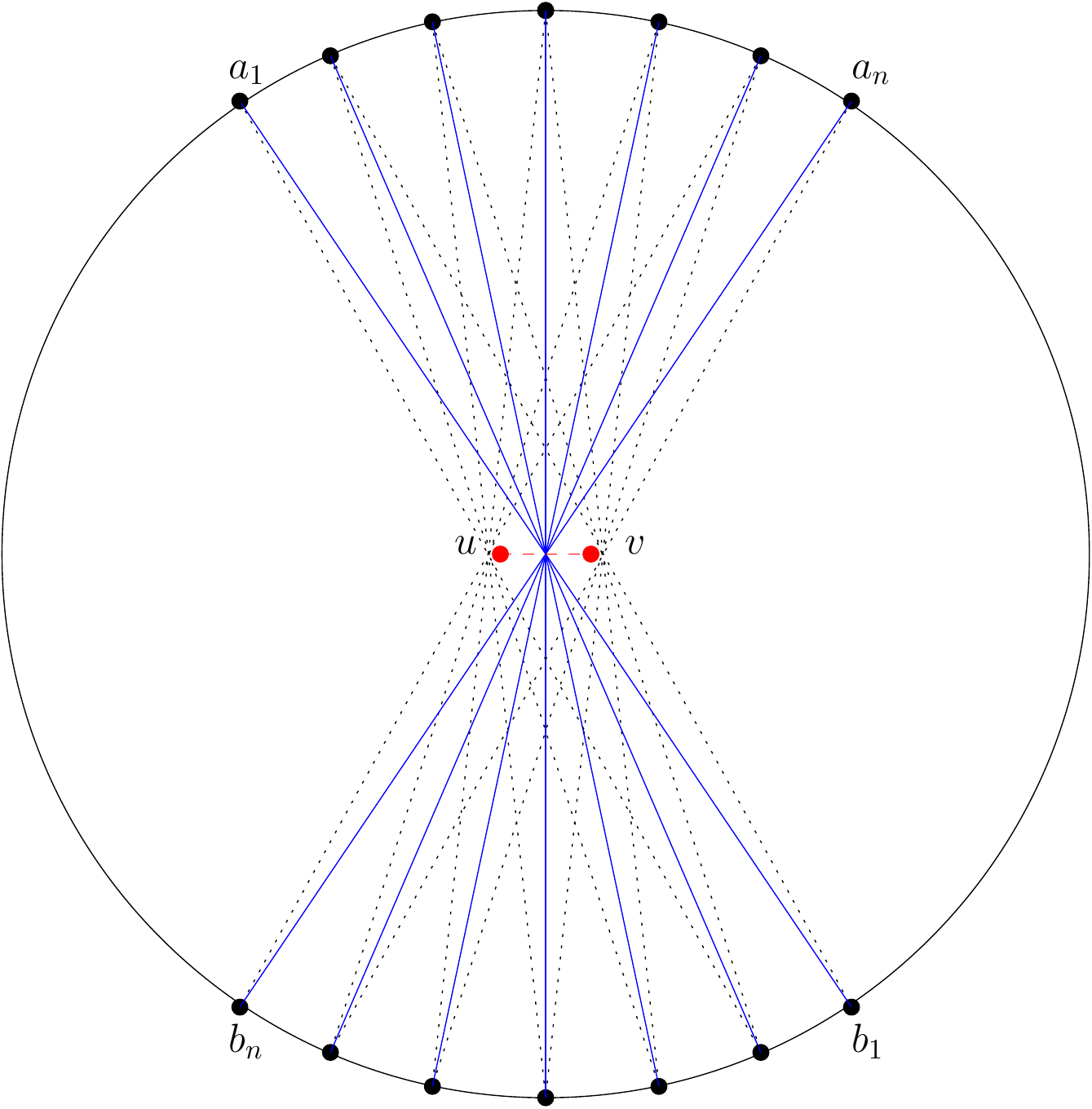}
    \caption{In this ``hourglass'' $uv$ is a flip cut edge that creates $n$ disconnected components in the flip graph, one for each $a_i b_i$.
}
    \label{fig:dis_comp}
\end{figure}

The ``hourglass'' shown in Figure~\ref{fig:dis_comp} has a flip cut edge $e=uv$ such that forbidding $e$ results in $\Theta(n)$ disconnected components in the flip graph, which is the most possible.  See Section~\ref{sec:hourglasses}.

Some point sets have $\Theta(n^2)$ flip cut edges but at the other extreme, 
some point sets have no flip cut edges. For example, there are no flip cut edges for points in convex position. 
The standard way to flip between two triangulations of a convex polygon is via a star triangulation with all edges incident to one point. If there is a point not incident to any forbidden edge, 
we can still flip to a star centered on that point.  See Figure~\ref{fig:flip-to-star}. A more detailed version of this argument is given in Lemma~\ref{lemma:F-touches-all} when we investigate  flip cut sets for points in convex position. 

\begin{figure}[ht]
    \centering
    \includegraphics[scale=0.7]{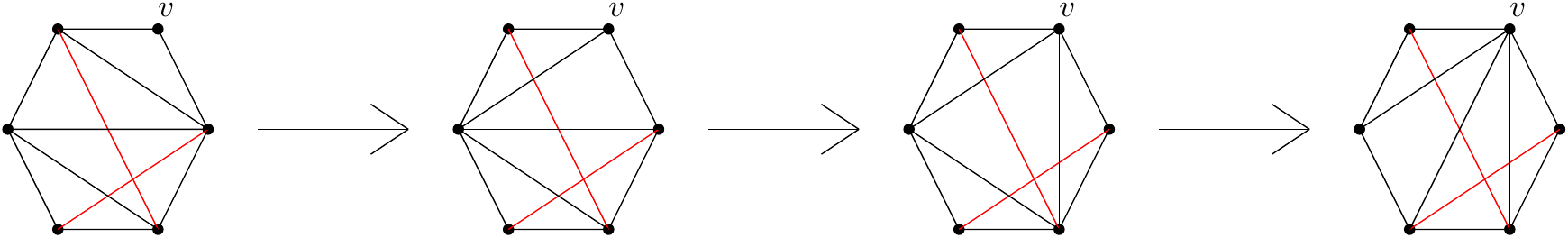}
    \caption{A convex hexagon with a set $X$ of two forbidden edges (in red), and a flip sequence to connect one triangulation in ${\cal T}_{-X}$ to the triangulation that is a star at $v$, without using any forbidden edges.
}
    \label{fig:flip-to-star}
\end{figure}

Figure~\ref{fig:PointSetFlipCutEdges} shows some more examples of flip cut edges in point sets. 

\begin{figure}[ht]
    \begin{subfigure}{0.49\textwidth}
        \centering
          \includegraphics[width=.9\linewidth]{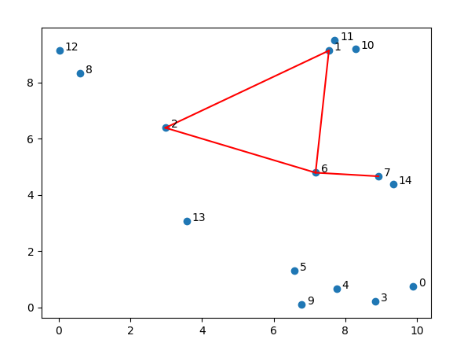}
          \label{fig:PointSetFlipCutEdges1}
    \end{subfigure}%
    \begin{subfigure}{0.49\textwidth}
          \centering
          \includegraphics[width=.9\linewidth]{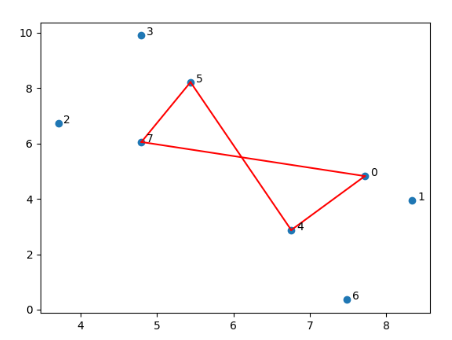}
          \label{fig:PointSetFlipCutEdges2}
    \end{subfigure}
    \begin{subfigure}{0.49\textwidth}
          \centering
          \includegraphics[width=.9\linewidth]{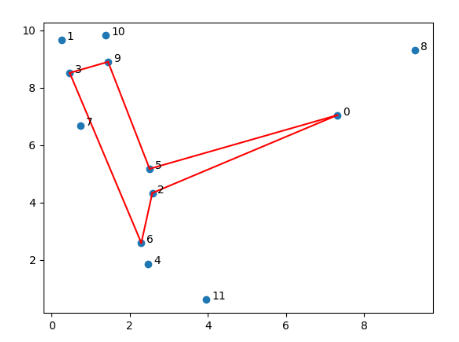}
          \label{fig:PointSetFlipCutEdges3}
    \end{subfigure}
    \begin{subfigure}{0.49\textwidth}
          \centering
          \includegraphics[width=.9\linewidth]{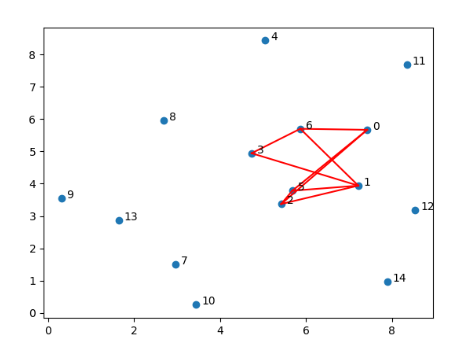}
          \label{fig:PointSetFlipCutEdges4}
    \end{subfigure}
    \caption{Some point sets and their flip cut edges (in red).
    }
    \label{fig:PointSetFlipCutEdges}
\end{figure}

\subsection{Notation and Definitions}

We denote the flip graph of point set $P$ by \defn{${\mathcal F}(P)$}, or just \defn{$\mathcal F$}, when $P$ is clear from the context.

For a subset $E$ of the edges of $P$, let \defn{${\cal T}_{+E}(P)$} be the set of triangulations of $P$ that include all the edges of $E$.
These are known as \defn{constrained triangulations}.  Let \defn{${\cal F}_{+E}(P)$} be the subgraph of the flip graph induced on the vertex set ${\cal T}_{+E}(P)$.  It is known that ${\cal F}_{+E}(P)$ is connected \cite{chew1989constrained}.

Let \defn{${\cal T}_{-E}(P)$} be the set of triangulations of $P$ that include \emph{none} of the edges of $E$,   
and let \defn{${\cal F}_{-E}(P)$} be the subgraph of the flip graph induced on ${\cal T}_{-E}(P)$.
When $E$ consists of a single edge $e$, we will write ${\cal T}_{-e}(P)$, and so on.  Also, we will omit $P$ when the point set is clear from the context.

A subset $E$ of edges of $P$ is a \defn{flip cut set} if the flip graph $\mathcal{F}_{-E}(P)$ is disconnected.  
The smallest size of a flip cut set is called the \defn{flip cut number} of $P$.
If $\{e\}$ is a flip cut set of size one, we call $e$ a \defn{flip cut edge}.

An \defn{empty convex quadrilateral} or \defn{EC4} is a set of 4 points in $P$ that form a convex quadrilateral with no other points of $P$ inside or on the boundary. We also use \defn{EC3} for empty triangles, \defn{EC5} for empty convex pentagons, etc. 

Point set $P$ is a \defn{convex point set} if all the points of $P$ are extreme points of the convex hull of $P$, i.e., for every point $p \in P$, there is a line through $p$ with all other points of $P$ strictly to one side of the line.
Point set $P$ is in \defn{general position} if no three points are collinear.

\section {Flip Cut Edges}
\label{sec:flip-cut-edges}

In this section we characterize flip cut edges (Section~\ref{sec:characterization}) and give an 
$O(n \log n)$
time algorithm to test if a given edge is a flip cut edge (Section~\ref{sec:flip-cut-edge-algorithm}).  In Section~\ref{sec:some-bounds} we examine some special point sets and establish bounds on the number of flip cut edges and the number of connected components.

\subsection{Characterizing Flip Cut Edges}
\label{sec:characterization}

Consider an edge $e=uv$ of $P$.  Let us orient $e$ horizontally so we can use the terms ``above'' and ``below'' to refer to the two sides of $e$. In the horizontal orientation, suppose that $u$ lies to the left of $v$.

Let $Y$ be the set of edges $f$ of $P$ such that $f$ crosses $e$.
Let $G_Y$ be the 
line graph of $Y$, i.e.,  the vertex set of $G_Y$ is $Y$, 
and two edges of $Y$ are adjacent in $G_Y$ if they meet at a point See Figure \ref{fig:Y}.

\begin{figure}[ht]
    \centering
    \includegraphics[scale=0.7]{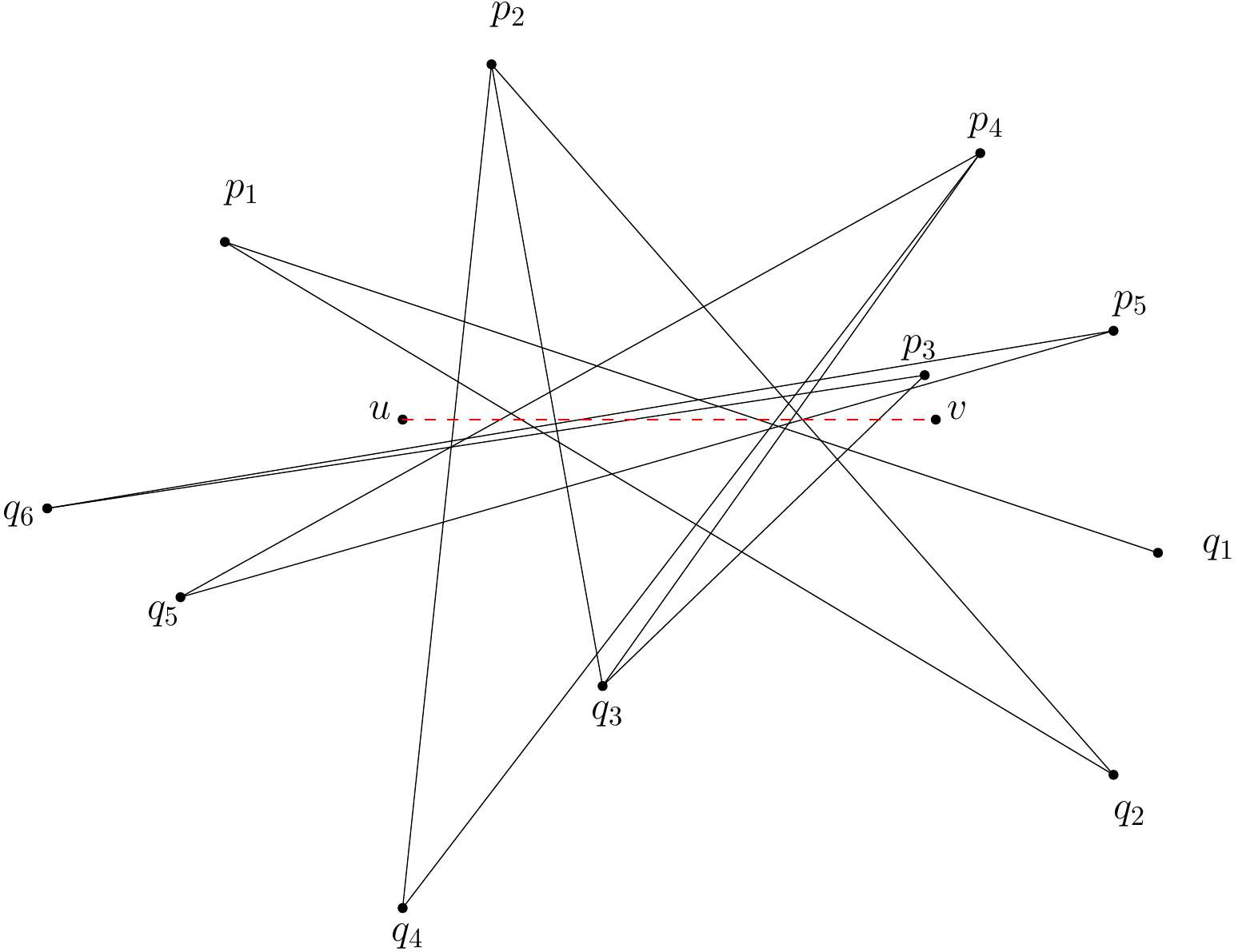}
    \caption{The edges $Y$ that cross $e = uv$.  Here, $G_Y$ is connected, so $e$ is not a flip cut edge.}
    \label{fig:Y}
\end{figure}

We will prove that $e$ is a flip cut edge if and only if $G_Y$ is disconnected. In fact, we will be able to identify the connected components of ${\cal F}_{-e}$ from $G_Y$.

\begin{observation}
\label{obs:union}
${\cal T}_{-e}$ is the union of the sets of triangulations ${\cal T}_{+f}$ for $f \in Y$.  Each flip graph ${\cal F}_{+f}$ is connected.
\end{observation}
\begin{proof}
For any edge $f$, ${\cal F}_{+f}$ is connected by properties of constrained triangulations~\cite{chew1989constrained}. 
For $f \in Y$, ${\cal T}_{+f}$ is a subset of ${\cal T}_{-e}$ because triangulations that contain $f$ cannot contain $e$.  Finally, since every triangulation of ${\cal T}_{-e}$  contains some edge of $Y$, ${\cal T}_{-e}$ is the union of ${\cal T}_{+f}$ for $f \in Y$.
\end{proof}

The observation says that every vertex of the graph ${\cal F}_{-e}$ appears in some ${\cal F}_{+f}$ for $f \in Y$.  In fact, every edge of the graph ${\cal F}_{-e}$ appears in some ${\cal F}_{+f}$ for $f \in Y$, as we will prove as part of Theorem~\ref{thm:new-characterization-2}.

Based on Observation~\ref{obs:union}, in order to identify connected components of the flip graph ${\cal F}_{-e}$, it suffices to figure out which subgraphs 
${\cal F}_{+f}$ are connected to which other ones in ${\cal F}_{-e}$, i.e., to know when there is a path in ${\cal F}_{-e}$ from an element of ${\cal T}_{+f}$ to an element of ${\cal T}_{+g}$.

Before giving the main theorem, we give one more observation.

\begin{observation}
\label{obs:Y-int-T-connected}
For any triangulation $T$ in ${\cal T}_{-e}$, the edges of $Y$ in $T$, which we denote $Y \cap T$, are connected in $G_Y$.  
\end{observation}
\begin{proof}
The triangles of $T$ that cross the segment $uv$ form 
a path in the planar dual of the triangulation, and so the edges of $Y \cap T$ form
a connected subgraph of $G_Y$.
See Figure \ref{fig:theorem4}.
\end{proof}

\begin{figure}[ht]
    \centering
    \includegraphics[scale=0.6]{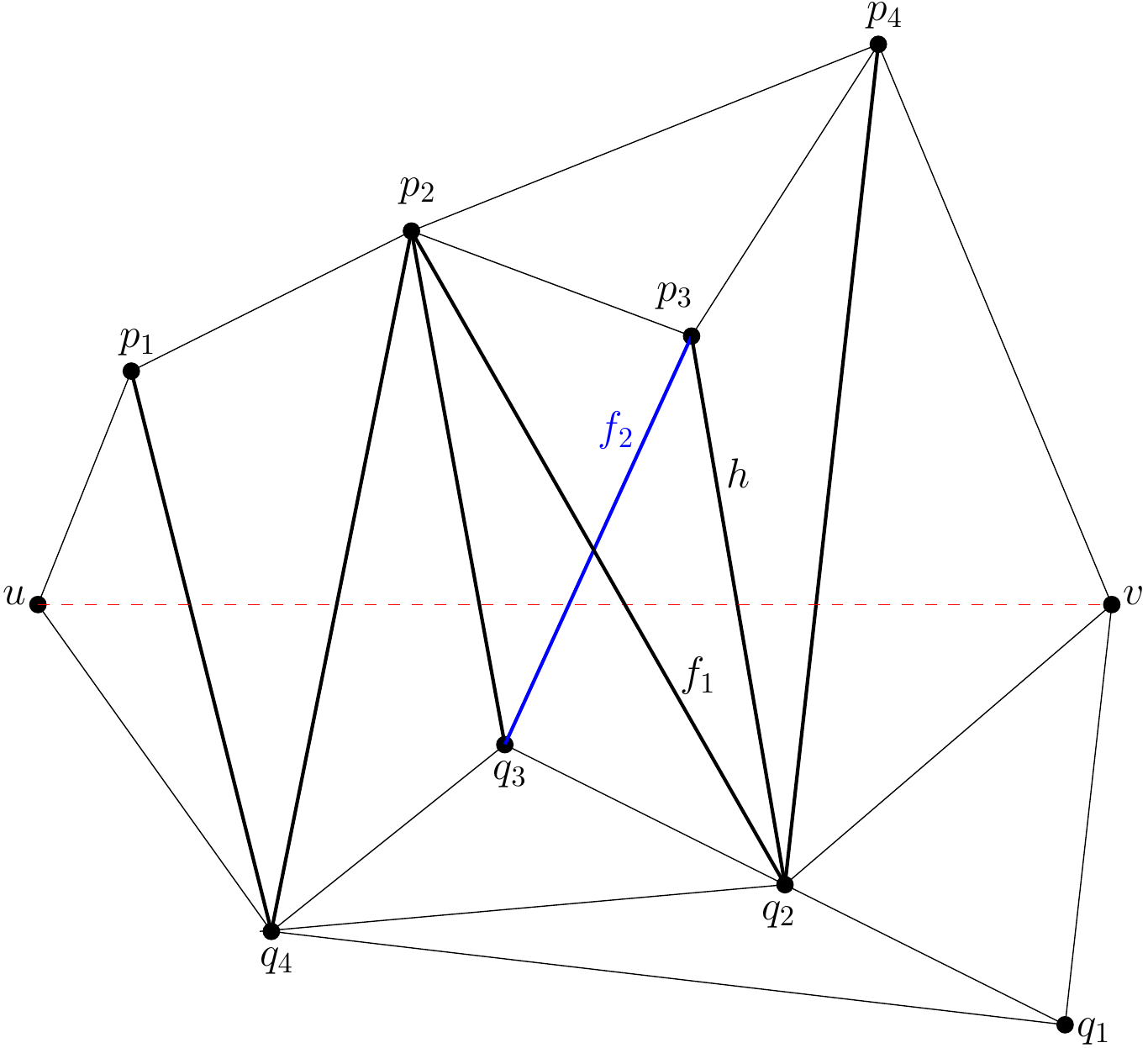}
    \caption{The edges of $Y \cap T$ (in thick black)
    are connected in $G_Y$.
    A flip of $f_1$ to $f_2$ requires another edge $h \in Y \cap T_1 \cap T_2$.
    }
    \label{fig:theorem4}
\end{figure}

\begin{theorem}
Edge $e$ is a flip cut edge if and only if $G_Y$ is disconnected.
More specifically, edges $f$ and $g$ in $Y$ are connected in $G_Y$ if and only if ${\cal T}_{+f}$ and  ${\cal T}_{+g}$ are connected in ${\cal F}_{-e}$.
\label{thm:new-characterization-2}
\end{theorem}

\begin{proof}
If $fg$ is an edge of $G_Y$, then, in the point set, the edges $f$ and $g$ are incident at a common endpoint so they do not cross, which implies that there is a triangulation containing $f$ and $g$, so ${\cal F}_{+f}$ and ${\cal F}_{+g}$ are connected (they have a triangulation in common).  Thus, by transitivity, if $f$ and $g$ are connected in $G_Y$ then ${\cal F}_{+f}$ and ${\cal F}_{+g}$ are connected.  

For the converse, it suffices to show that 
if we flip from $T_1$ in ${\cal T}_{-e}$ to $T_2$ in ${\cal T}_{-e}$, then $Y \cap T_1$ and $Y \cap T_2$ are connected in $G_Y$. 
First note that by Observation~\ref{obs:Y-int-T-connected}, $Y \cap T_1$ is a connected set in $G_Y$ and the same for $Y \cap T_2$.
We will 
show that $Y \cap T_1$ and $Y \cap T_2$ have an element in common. Let $f_i \in T_i \cap Y$, $i=1,2$. If either edge is in the other triangulation, we are done.  Otherwise, the flip from $T_1$ to $T_2$ must flip $f_1$ to $f_2$.  Then the EC4 formed by the endpoints of $f_1$ and $f_2$ has a side edge $h$ that is an edge of $Y$ in both $T_1$ and $T_2$.  
See Figure \ref{fig:theorem4}.
(Note that this argument shows that every edge of the flip graph ${\cal F}_{-e}$ lies in  ${\cal F}_{+h}$ for some $h \in Y$.) 
\end{proof}

\medskip\noindent{\bf Algorithmic implications.} The theorem gives an immediate polynomial time algorithm to test if $e$ is a flip cut edge: construct $G_Y$ (a graph on $O(n^2)$ vertices) and test connectivity. 
In Subsection~\ref{sec:flip-cut-edge-algorithm} we give a faster  $O(n \log n)$ time algorithm. 

Our other algorithmic goal is to ``identify'' the connected components of ${\cal F}_{-e}$.  Since the flip graph is exponentially large, we focus on the problem of testing whether two triangulations $T_1$ and $T_2$ in ${\cal T}_{-e}$ are in the same connected component of ${\cal F}_{-e}$.  Using Theorem~\ref{thm:new-characterization-2} we can do that as follows.  Pick $f_1$ in $Y \cap T_1$ and $f_2$ in $Y \cap T_2$ (note that such edges exist).  Then $T_1 \in {\cal T}_{+f_1}$ and $T_2 \in {\cal T}_{+f_2}$.  So $T_1$ and $T_2$ are connected in ${\cal F}_{-e}$ iff ${\cal T}_{+f_1}$ and ${\cal T}_{+f_2}$ are connected in  ${\cal F}_{-e}$ iff (by Theorem~\ref{thm:new-characterization-2}) $f_1$ and $f_2$ are connected in $G_Y$, which we can test in polynomial time. 
In Subsection~\ref{sec:flip-cut-edge-algorithm} we give a faster  $O(n)$ time algorithm. 

\medskip\noindent{\bf Alternative characterization of flip cut edges.} 
For the efficient algorithms in Section~\ref{sec:flip-cut-edge-algorithm}, 
we 
need an alternative characterization of flip cut edges in terms of a subgraph of $G_Y$.

For any $f \in Y$, let $Q(f)$ be the convex quadrilateral formed by the endpoints of $f$ and $e=uv$.  
Let $Z$ be the set of edges $f \in Y$ such that $Q(f)$ is an EC4.  
Let $G_Z$ be the line graph of $Z$.

Note that any $f \in Z$ has one point above $e$ and one point below $e$, and these points make empty triangles with $e$. 
Let $A$ be the set of points $a$ above $e$ such that $auv$ is an empty triangle (an EC3).  Let $B$ be the set of points $b$ below $e$ such that $bvu$ is an empty triangle.  Thus $Z$ consist of the edges from $A$ to $B$ that cross $e=uv$, i.e.,  $Z = Y \cap (A \times B)$.
See Figure~\ref{fig:order}.

We will prove an analogue of Theorem~\ref{thm:new-characterization-2}:

\begin{theorem}
\label{thm:Z-characterization}
Edge $e$ is a flip cut edge if and only if $G_Z$ is disconnected.    
\end{theorem}

We will also be able to 
characterize connected components of ${\cal F}_{-e}$ in terms of $G_Z$, but this 
cannot be analogous to Theorem~\ref{thm:new-characterization-2},
because not every triangulation contains an edge of $Z$.
We begin with some preliminary results.

\begin{lemma}
\label{lemma:G_ZconnectedG_Y}
Edges $f,g \in Z$ are connected in $G_Z$ if and only if they are connected in $G_Y$.
\end{lemma}  

\begin{proof}
The forward direction 
is clear since $G_Z$ is an induced subraph of $G_Y$.

For the other direction, suppose $f$ and $g$ are connected in $G_Y$. 
Suppose, $f=f_1,f_2,...,$ $f_m=g$ is a shortest path in $G_Y$ between $f$ and $g$. If all the $f_i$'s are in $Z$, then $f$ and $g$ are connected in $G_Z$.
Otherwise, we will modify the path to replace edges of $Y$ by edges of $Z$.
Suppose $f_i = p_i q_i$ with $p_i$ above $uv$ and $q_i$ below $uv$. 
Because the path is shortest, the common points between successive edges are alternately above and below the line $L$ through $uv$.  In particular, suppose without loss of generality that $p_1 = p_2, q_2 = q_3, \ldots , p_{2i-1} = p_{2i}, q_{2i} = q_{2i+1}, \ldots$. 

The plan is to replace $p_i$ by some point $p'_i$ and replace $q_i$ by some point $q'_i$ so that the segments $f'_i = p'_i q'_i$ are edges in $Z$ and form a path connecting $f$ to $g$.

\begin{figure}[ht]
    \centering
    \includegraphics[scale=0.6]{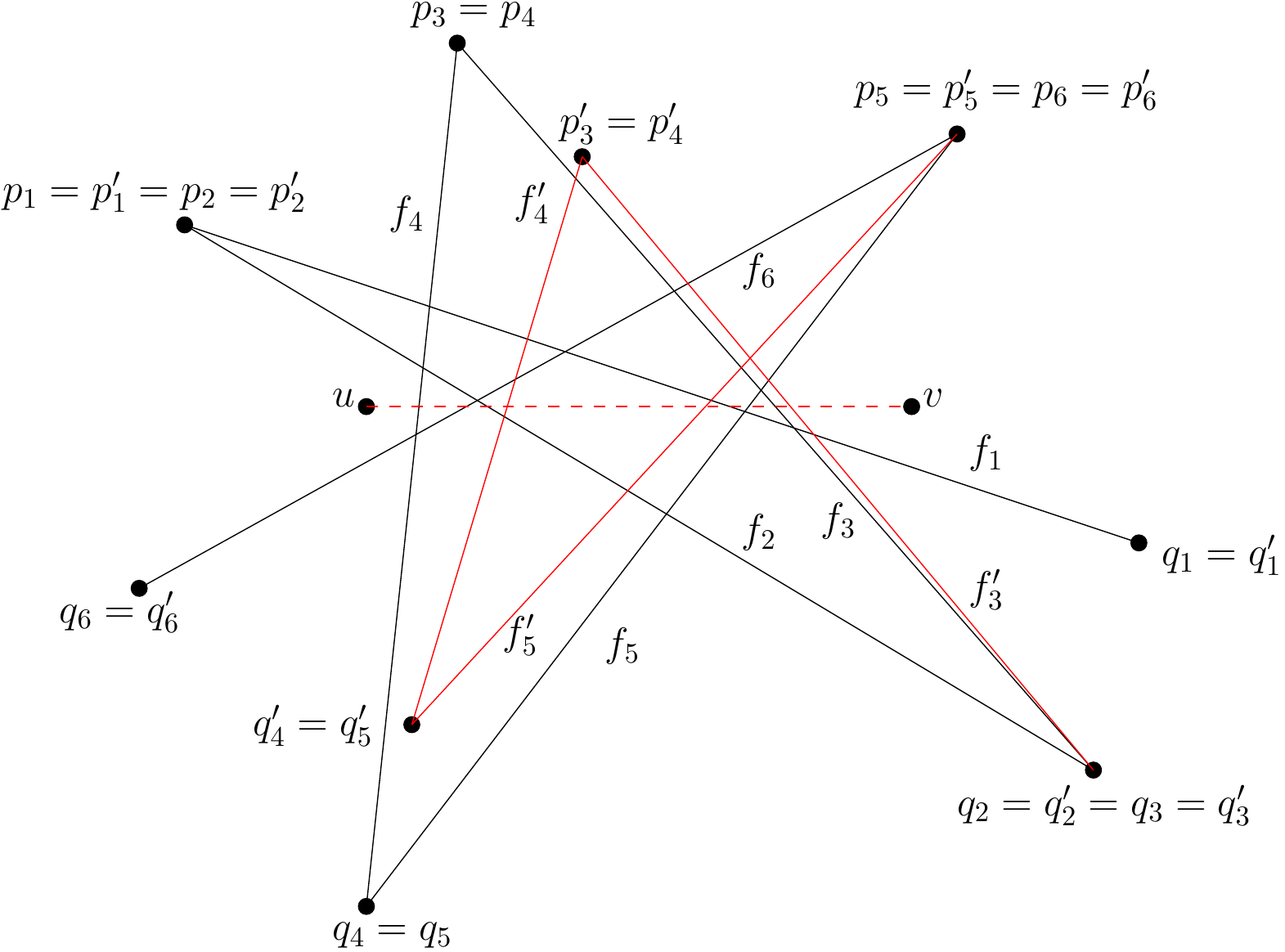}
    \caption{Proof of Lemma \ref{lemma:G_ZconnectedG_Y}.
    The original path in $G_Y$ from $f_1 \in Z$ to $f_6 \in Z$ is $f_1, \ldots, f_6$ (shown in black), and the modified path replaces the edges of $Y-Z$, which are $f_3, f_4, f_5$ by the edges $f'_3, f'_4, f'_5 \in Z$  (shown in red).
    }
    \label{fig:lemma6}
\end{figure}

If triangle $p_i uv$ is empty, then $p'_i = p_i$; otherwise 
$p'_i$ is a point inside triangle $p_i uv$ that is closest to line $L$.  Define $q'_i$ similarly with respect to triangle $q_i uv$.
Note that $f'_1 = f_1$ and $f'_m = f_m$.

We claim that each $f'_i$ is in $Z$.  We must show that the line segment $p'_i q'_i$ is an edge, that it crosses $e=uv$, and that the quadrilateral $Q(f'_i)$ is empty.

First note that triangles $p'_i uv$ and $q'_i uv$ are empty because we picked $p'_i$ and $q'_i$ closest to line $L$.  Next we claim that these two empty triangles together form a convex quadrilateral.  This is because $f_i$ crosses $uv$, so $Q(f_i)$ is convex, and when we move $p_i$ to $p'_i$ and $q_i$ to $q'_i$, convexity is preserved.  Thus $f'_i$ is an edge of $Z$.

Finally, note that $p'_{2i-1} = p'_{2i}$ and $q'_{2i} = q'_{2i+1}$ because that was true of the original points.
Thus the edges $f'_i$ form a path in $G_Z$ connecting $f$ and $g$.
\end{proof}

\begin{lemma}
\label{lemma:YcontainZ}
Every connected component of $G_Y$ contains an edge of $Z$.
\end{lemma}

This lemma can be proved in several different ways.  One possibility is to 
take an edge $f \in Y$ and construct the edge $f' \in Z$ inside $Q(f)$ as in the above proof. 
However, showing that $f$ and $f'$ are connected in $G_Y$ runs into some complications due to the possibility of collinear points. 

Instead we take an edge $f \in Y$ and consider a triangulation $T$ containing $f$.  
Below we give a short proof based on this idea, and in Section~\ref{sec:flip-cut-edge-algorithm} we give an alternative algorithmic proof that efficiently finds an edge of $Z$ connected to $f$ in $G_Y$. 

\begin{proof}
Let $f$ be an edge of $Y$
and let $T$ be a triangulation that contains $f$.  
We will prove there is an edge $g \in Z$ that is connected to $f$ in $G_Y$. 
There is a flip sequence from $T$ to a triangulation that contains $e=uv$. The last flip in this sequence must involve an edge $g$ of $Z$ flipping to $e$.
Until the last flip, we are in one connected component of ${\cal F}_{-e}$, so, by Theorem~\ref{thm:new-characterization-2}, all the edges of $Y$ in all the trianglulations in the sequence are in the same connected component of $G_Y$.  Thus $f$ and $g$ are connected in $G_Y$.
\end{proof}

\begin{proof}[Proof of Theorem~\ref{thm:Z-characterization}]
If $e$ is a flip cut edge, then by Theorem~\ref{thm:new-characterization-2}, $G_Y$ is disconnected.  By Lemma~\ref{lemma:YcontainZ}, every connected component of $G_Y$ contains an edge of $G_Z$.  Thus $G_Z$ is disconnected.

In the other direction, if $G_Z$ is disconnected, then by Lemma~\ref{lemma:G_ZconnectedG_Y}, $G_Y$ is disconnected, so $e$ is a flip cut edge by Theorem~\ref{thm:new-characterization-2}.
\end{proof}

In the following section we give efficient algorithm that uses $G_Z$ to test if edge $e$ is a flip cut edge, and, if so, to identify when two given triangulations are in different components of ${\cal F}_{-e}$.

\begin{figure}[ht]
    \centering
    \includegraphics[width=.6\textwidth ]{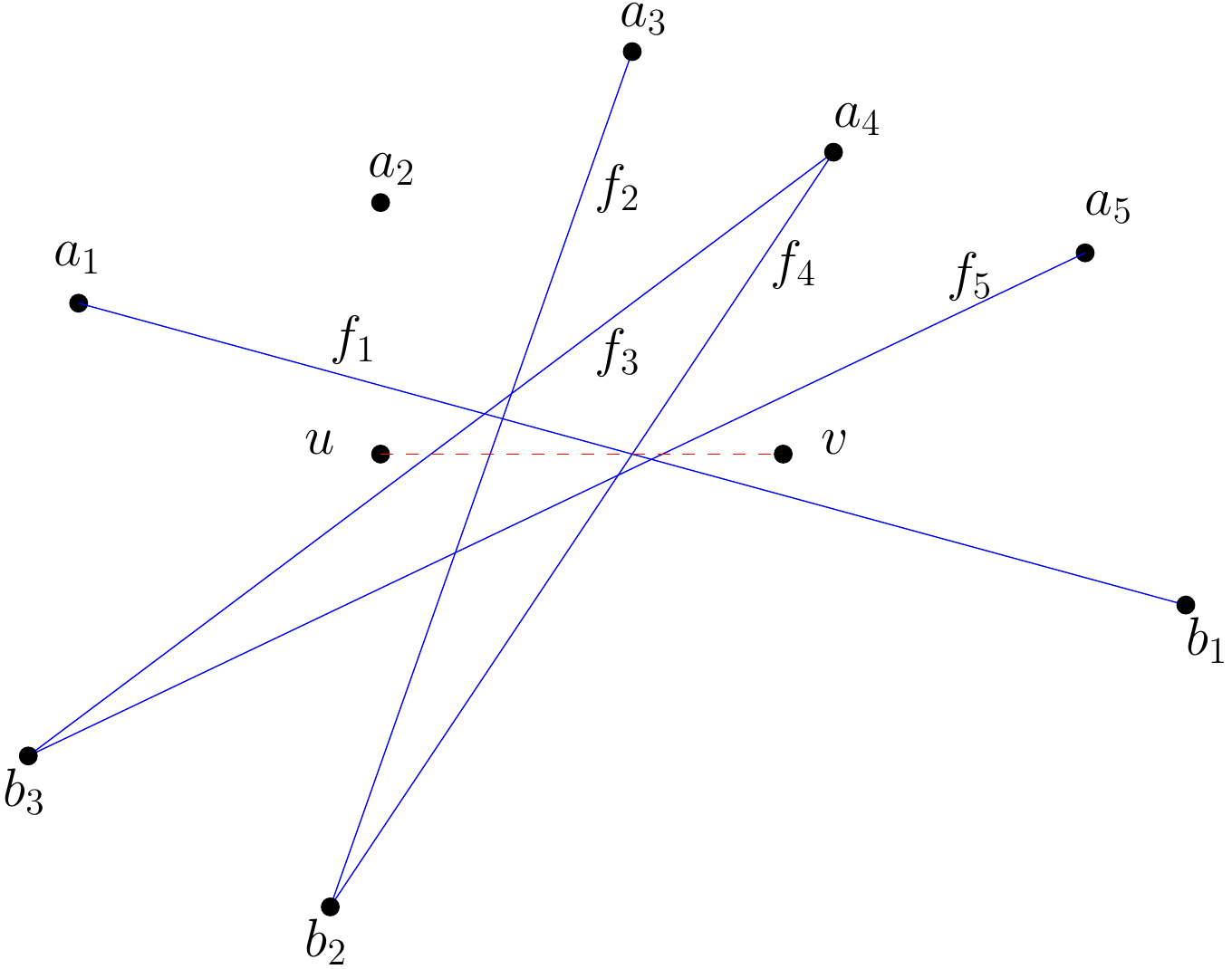}
    \caption{The points that form empty triangles with $uv$ are $A = a_1,  \ldots, a_5$ and $B = b_1, b_2, b_3$ ordered cyclically around $u$. The edges of $Z$ are $f_1, \ldots, f_5$ and they form two connected components in $G_Z$, $\{f_1\}$ and $\{f_2,\ldots ,f_5\}$. 
For any point $b \in B$, the points $a \in A$ such that $ab \in Z$ form a subinterval of $A$.  }
    \label{fig:order}
\end{figure}

\subsection{Algorithms for Flip Cut Edges}
\label{sec:flip-cut-edge-algorithm}

In this section we give an $O(n \log n)$ time algorithm to test if a given edge $e$ is a flip cut edge.  
Then, with that algorithm as a preprocessing step, 
we give an $O(n)$ time algorithm to test if two triangulations in ${\cal T}_{-e}$ are connected in ${\cal F}_{-e}$.
We also show how to find all flip cut edges in time $O(n^3)$ (note that there may be $\Theta(n^2)$ flip cut edges). 

\subsubsection{Testing for a Flip Cut Edge}
To test if edge $e$ is a flip cut edge, 
it suffices to test if $G_Z$ is connected by Theorem~\ref{thm:Z-characterization}.  
The algorithm also identifies  the connected components of $G_Z$ (though without explicitly listing the elements of $Z$, since there can be $\Theta(n^2)$ of them).
In particular, we find disjoint subsets $A_1, \ldots , A_c$ of $A$, and $B_1, \ldots, B_c$ of $B$ such that the $i$th connected component of $G_Z$ consists of the edges of $Z$ between $A_i$ and $B_i$.  
Note that $Z$ will not in general contain all pairs from $A_i \times B_i$.

We first need  some more properties of the sets $A$ and $B$ and the graph $G_Z$.
Because no empty triangle is contained in another, 
the ordering of $a \in A$ by decreasing (convex) angle $\angle auv$ is the same as the ordering by increasing (convex) angle $\angle avu$. 
Let $a_1, \ldots , a_k$ be this ordering of $A$.  Similarly, let $b_1, \ldots, b_l$ be the ordering of $B$ by decreasing (convex) angle $\angle bvu$, or equivalently, by increasing angle $\angle buv$.  Thus the cyclic order of $A \cup B$ around $u$ is $a_1, \ldots, a_k, b_1, \ldots, b_k$.  
See Figure~\ref{fig:order}.

\begin{observation}
For any point $b_j$, the set of points $a_i$ such that the edge $a_i b_j$ is in $Z(e)$ form a subinterval of the ordering of $A$.  And similarly for $a_i$.
\label{obs:consecutive-neighbours}
\end{observation}

Later on, we will find it useful to have an even stronger property:

\begin{observation}
\label{obs:Z-edges}
Let $L$ be the line through $uv$.
\begin{enumerate}
\item If $a_i b_j$ crosses $L$ to the right of $v$ then the same is true for 
all $a_{i'} b_{j'}$, $i' \ge i, j' \le j$.
\item If $a_i b_j$ crosses $L$ to the left of $u$ then the same is true for 
all $a_{i'} b_{j'}$, $i' \le i, j' \ge j$.
\item 
If $a_{i_1} b_{j_2}$ and $a_{i_2} b_{j_1}$ are in $Z$ for some $i_i \le i_2$ and $j_1 \le j_2$ then $a_i b_j$ is in $Z$ for all $i_1 \le i \le i_2$ and all $j_1 \le j \le j_2$.
\end{enumerate}
\end{observation}
\begin{proof} It suffices to prove (1), since (2) is symmetric, and (3) follows from (1) and (2).
To prove (1), suppose $a_i b_j$ crosses $L$ to the right of $v$.  Then the angle $\angle a_i v b_j$ is convex to the right of $v$, so (by the ordering) the same is true of $\angle a_{i+1} v b_j$ and $\angle a_i b_{j-1}$.  The result follows by induction. 
\end{proof}

\medskip\noindent{\bf Finding the ordered sets $A$ and $B$.}
We show how to find the ordered set $A$ in $O(n \log n)$ time. 
Sort all the points $p \in P$ that lie above $e$ by decreasing angle $\angle puv$, breaking ties by distance from $u$.  Call this the \emph{$u$-order}, $<_u$.
Also sort the same points by increasing angle $\angle pvu$, and call this the \emph{$v$-order}, $<_v$.
Define $a_1$ be 
the first point in the $v$-order.  Observe that $a_1$ is a member of $A$ (the first member of $A$), and that any points $p$ with $p <_u a_1$ do NOT belong to $A$ and can be discarded.  
In general, we proceed through the $v$-order.  Define $a_i$ to be the next un-discarded element of the $v$-order, and discard all points $p$ with $p <_u a_i$. 

To prove that this is correct, observe that when we discard $p$,
we have $p <_u a_i$ and $p >_v a_i$, so the triangle 
$puv$ contains $a_i$, so discarding $p$ is correct. And observe that when we choose $a_i$, we have $a_{i-1} <_v a_{i}$ because we follow the $v$-order, and we have $a_{i-1} <_u a_i$ because we discarded all points $p <_u a_{i-1}$ in the previous step.

We can find the ordered set $B$ similarly. The time to find the ordered sets $A$ and $B$ is $O(n \log n)$, and in fact, it is $O(n)$ after the preliminary sorting steps.

\medskip\noindent{\bf Finding the connected components of $G_Z$.}
Given the ordered lists $A$ and $B$, we find the connected components of $G_Z$ in linear time as follows. 
The algorithm has two phases.  In Phase 1, we find the first edge of the next component, and in Phase 2, we complete the component.
Phase 1 initially begins with 
$a_1$ and $b_1$, but more generally, we will find the first edge of the next component $A_c, B_c$ among the ``active'' points $a_i, \ldots, a_k$ and $b_j, \ldots, b_\ell$, maintaining the invariant that there are no edges of $Z$ from an inactive point to an active point. We search for an edge that crosses $uv$.  Let $L$ be the line through $uv$.  If $a_i b_j$ crosses $L$ to the right of $v$, then
by Observation~\ref{obs:Z-edges}(1), there are no edges of $Z$ from $b_j$ to any active point.
So we increment $j$ and the invariant is maintained. 
If $a_i b_j$ crosses $L$ to the left of $u$, then, with similar justification, we increment $i$. 
Continue until $a_i b_j$ crosses $uv$. This completes Phase 1---we have found the first edge of the connected component.  Add $a_i$ to $A_c$ and add $b_j$ to $B_c$.

For Phase 2, we complete the connected component by
alternately ``pivoting'' on $a_i$ and $b_j$.  To pivot on $a_i$, update $j$ to the maximum index such that $a_i b_j$ crosses $uv$, putting all the $b$ points we find into $B_c$.  To pivot on $b_j$, update $i$ to the maximum index such that $a_i b_j$ crosses $uv$, putting all the $a$ points we find into $A_c$. 
When no more pivots are possible, we are done with this connected component and done with Phase 2.  We increment $i$ and $j$ by 1, and go back to Phase 1 to find the next connected component.

More details are given as Algorithm~\ref{algorithm:selection}.  The algorithm runs in linear time since it only performs a linear scan through each of the ordered lists $A$ and $B$, doing constant work for each point.

We now justify correctness of the algorithm.  We already argued that during Phase 1 we maintain the invariant that there are
no edges of $Z$ between an inactive point and an active point, i.e., 
there are no edges of $Z$ of the form $a_{i'} b_{j'}$ with $i' \le i$ and $j' >j$ or with $i' > i$ and $j' \le j$.  
Now consider Phase 2.
Observe that when we add points to $A_c$ and $B_c$, they are part of the same connected component of $G_Z$.  This is because when we pivot on $a_i$, all the $b$ points that we add to $B_c$ are adjacent to $a_i$ in $Z$, and similarly, when we pivot on $b_j$, all the $a$ points that we add to $A_c$ are adjacent to $b_j$ in $Z$.  It remains to show that when we declare a connected component ``done'' because we can no longer pivot at $a_i$ or $ b_j$, 
then the invariant holds for active points $a_{i+1}, \ldots, a_k$ and $b_{j+1}, \ldots, b_\ell$, 
i.e.,  there are no edges of $Z$ of the form $a_{i'} b_{j'}$ with $i' \le i$ and $j' >j$ or with $i' > i$ and $j' \le j$. First suppose $a_{i'}b_{j'} \in Z$ for some $i' \le i$ and $j' >j$.  
Then by Observation~\ref{obs:Z-edges}(3), $a_i b_{j+1}$ is in $Z$, so we would not have been done pivoting at $a_i$. 
The case of $i' > i$ and $j' \le j$ follows by symmetry.

\begin{algorithm}[htbp]
    \caption{Connected components of $G_Z$}
    \label{algorithm:selection}
\hspace*{\algorithmicindent} \textbf{Input:} $u,v, a_1, \ldots , a_k, b_1, \ldots, b_\ell$\\
\hspace*{\algorithmicindent} \textbf{Output:} Connected components of $G_Z$ in the form of $A_1, \ldots, A_c, B_1, \ldots, B_c$
\begin{algorithmic}[1]
\Procedure{ConnctedComponents}{}
\State for all $i$, 
{\tt done[$a_i$]} $\leftarrow$ False; 
for all $j$, 
{\tt done[$b_j$]} $\leftarrow$ False 
\State $i \leftarrow 1$;   $j \leftarrow 1$; $c \leftarrow 0$

\While{$i \le k$ and $j \le \ell$} 
    \State {\color{blue} // Phase 1. Find the first edge of  component $c$}
    
    \While {$a_i b_j$ does not cross $uv$}
        \If{$a_ib_j$ crosses $L$ to the right of $v$} 
            \State $j \leftarrow j+1$; {\bf if}
            $j > \ell$ {\bf then Halt}
        \EndIf
        \If{$a_ib_j$ crosses $L$ to the left of $u$} 
            \State $i \leftarrow i+1$; {\bf if}
            {$i > k$} {\bf then Halt}
        \EndIf
    \EndWhile

    \State $c \leftarrow c + 1$;
Insert $a_i$ into $A_c$, and $b_j$ into $B_c$
    \State {\color{blue}// Phase 2. Find all edges in  component $c$ by alternately pivoting on $a_i$, $b_j$ until no further pivot is possible}
    \Repeat 
        \State {\color{blue} // Pivot on $a_i$}
        \While{$a_ib_j$ crosses $uv$ and $j \le \ell$}
            \State insert $b_j$ into $B_c$;
            $j \leftarrow j+1$
        \EndWhile
        \State $j \leftarrow j - 1$;
        {\tt done[$a_i$] $\leftarrow$ True}
        \State {\color{blue} // Pivot on $b_j$}
        \While{$a_ib_j$ crosses $uv$ and $i \le k$}
            \State insert $a_i$ into $A_c$;
            $i \leftarrow i+1$
        \EndWhile
         \State $i \leftarrow i - 1$;
        {\tt done[$b_j$] $\leftarrow$ True}
    \Until{{\tt done[$a_i$]} and {\tt done[$b_j$]}}
    \State{$i \leftarrow i+1$; $j \leftarrow j+1$}
\EndWhile

\EndProcedure
\end{algorithmic}
\end{algorithm}

\subsubsection{Finding All Flip Cut Edges}

To find all the flip cut edges, we run the above test on each of the $O(n^2)$ edges.  We preprocess by sorting the points cyclically around each point $p$ in a total of $O(n^2 \log n)$ time.  Then, to test a particular edge $e = uv$, we break the orderings around $u$ and around $v$, and apply Algorithm~\ref{algorithm:selection}, which takes linear time apart from the sorting.  Thus the total time to find all flip cut edges is $O(n^3)$.

\subsubsection{Testing Connectivity of Two Triangulations}
We now show how to test if two triangulations $T_1$ and $T_2$ in ${\cal T}_{-e}$ are connected in ${\cal F}_{-e}$. We assume that we have the output of the above algorithm, i.e., the sets $A_1, \ldots, A_c$ and $B_1, \ldots, B_c$ such that the $i$th connected component of $G_Z$ consists of the edges of $Z$ between $A_i$ and $B_i$. 
We can then find the component $i$ of a given edge $e \in Z$ in constant time.

As mentioned in Section~\ref{sec:characterization} one approach is to pick one edge $f_1$ from $Y \cap T_1$, and one edge $f_2$ from $Y \cap T_2$, and test if $f_1$ and $f_2$ are in the same connected component of $G_Y$.  However, we only have $G_Z$ available to us, and there are triangulations in ${\cal T}_{-e}$ that contain no edges of $Z$.

Instead, we give an algorithmic version of Lemma~\ref{lemma:YcontainZ} that finds an edge $g_1 \in Z$ in the same component of $G_Y$ as the edges $Y \cap T_1$, and an edge $g_2 \in Z$ in the same component of $G_Y$ as the edges $Y \cap T_2$. 
Then we simply test if $g_1$ and $g_2$ are in the same set $A_i \times B_i$.

We first establish correctness and then give the details of finding $g_1$ and $g_2$ in linear time.  We must prove that $g_1$ and $g_2$ are connected in $G_Z$ if and only if $T_1$ and $T_2$ are connected in ${\cal F}_{-e}$.  Let $f_i$ be an edge of $Y \cap T_i$.  By Observation~\ref{obs:Y-int-T-connected}, the edges of $Y \cap T_i$ are all connected in $G_Y$, so the choice of $f_i$ is arbitrary.  
Now $T_1$ and $T_2$ are connected in ${\cal F}_{-e}$ iff (by Theorem~\ref{thm:new-characterization-2}) $f_1$ and $f_2$ are connected in $G_Y$ iff (by choice of $g_1$, $g_2$) $g_1$ and $g_2$ are connected in $G_Y$ iff (by Lemma~\ref{lemma:G_ZconnectedG_Y}) $g_1$ and $g_2$ are connected in $G_Z$.

We now give the details of finding $g_1$ (finding $g_2$ is exactly similar).  
Triangulation  $T_1$ has a 
sequence $C$ of triangles that intersect $uv$. 
See Figure~\ref{fig:triangles-along-e}.
Each triangle in $C$ shares an edge of $Y \cap T_1$ with the previous triangle in $C$, and all the edges of $Y \cap T_1$ are in one connected component of $G_Y$ (by Observation~\ref{obs:Y-int-T-connected}).  
Among the vertices of triangles of $C$, let $a$ be a vertex above the line $L$ through $uv$ that is closest to $L$ and let $b$ be  a vertex below the line $L$ that is closest to $L$. Then $ab$ is an edge of $Z$ and is in the same connected component as $Y \cap T_1$. 
This can be done in linear time.

\begin{figure}[ht]
    \centering
    \includegraphics[scale=0.6]{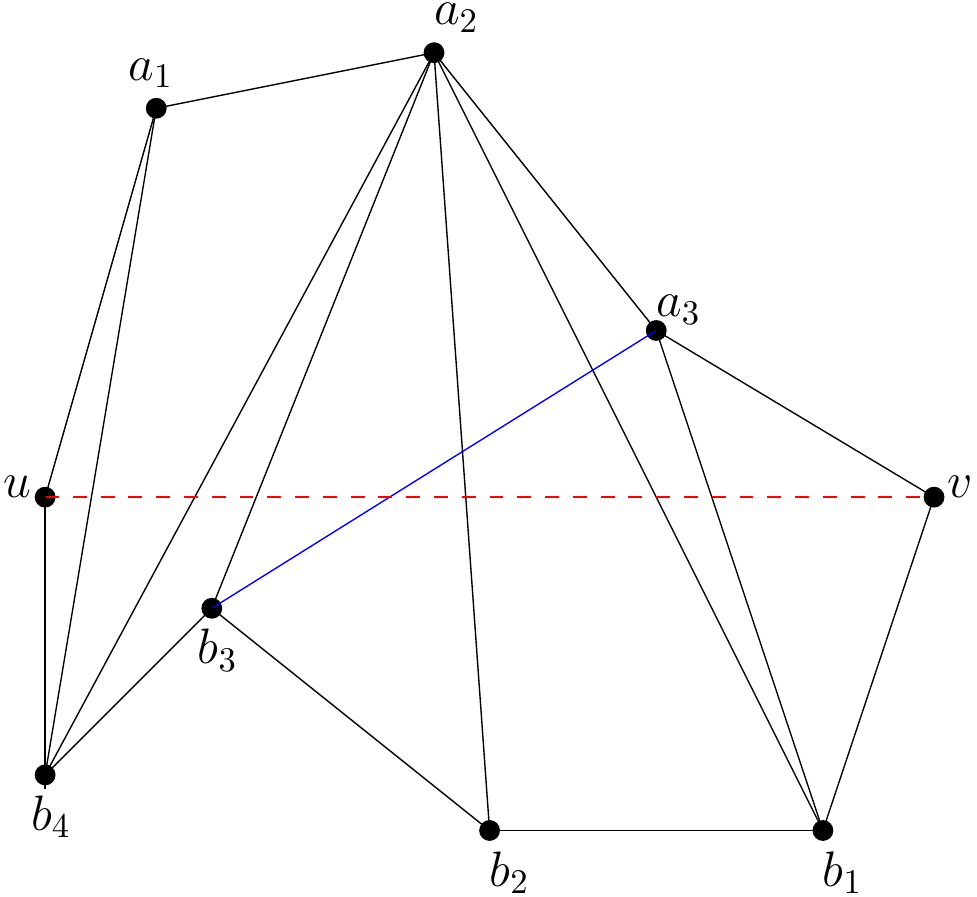}
    \caption{
    The triangles $C$ of $T_1$ that cross edge $e = uv$, 
and an edge $a_3 b_3$ of $Z$ (in blue) that is connected in $G_Y$ to the edges of $Y \cap T_1$. 
 }
    \label{fig:triangles-along-e}
\end{figure}

\subsection{Further Results on Flip Cut Edges}
\label{sec:some-bounds}

In this section we establish some bounds on the number of flip cut edges, and on the number of disconnected components caused by forbidding one flip cut edge.
To do this, we fill in the details of the claims made in Section~\ref{sec:examples} (Examples) about some special point sets: channels, hourglasses, and grid point sets.  We begin by relating flip cut edges to empty convex pentagons. 
Throughout the section we refer to a forbidden edge $e$, and to $Z$, $G_Z$, $A$, $B$, etc., as defined above.

\subsubsection{Empty Convex Pentagons (EC5's)}
\label{sec:no-EC5}

Empty convex pentagons (EC5's) play a significant role in the study of flip graphs.
Eppstein~\cite{eppstein} studied the flip graph of points without EC5's, and gave a polynomial time algorithm to find the minimum number of flips between two given triangulatioins.  Points on an $n \times m$ grid have no EC5's, and in the other direction, a point set without EC5's must have collinear points.  In particular,  any $n \ge 10$ points in general position contain an EC5~\cite{harborth1978konvexe}, and sufficiently large point sets with no EC5 have arbitrarily large subsets of collinear points~\cite{abel2011every}.
Empty convex pentagons also play a significant role in analyzing flips for coloured edges and proving the Orbit Theorem~\cite{lubiw2019orbit}.

We show that the absence of EC5's leads to a simple condition for flip cut edges.  In particular, if a point set has no EC5's, then 
edge $e$ is a flip cut edge if and only if $e$ is in (i.e., is a diagonal of) at least two EC4's. In fact it suffices to exclude local EC5's:

\begin{lemma}
\label{lemma:two-diagonals}
An edge $e$ that is not a diagonal of an EC5 is a flip cut edge if and only if it is a diagonal of at least two EC4's, i.e., $|Z| \ge 2$.    
\end{lemma}

\begin{figure}[ht]
    \centering
    \includegraphics[width=.4\textwidth ]{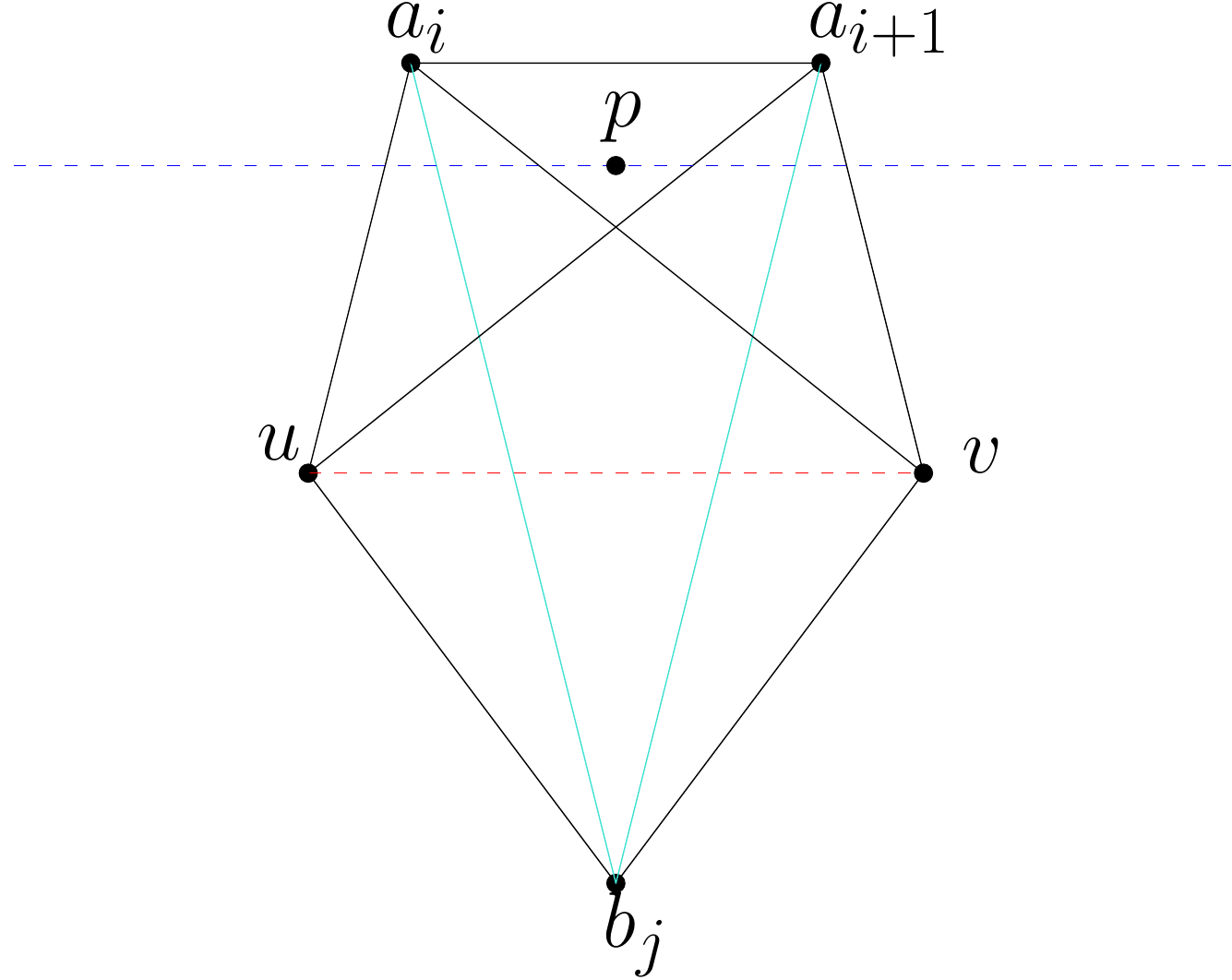}
    \caption{The endpoints of the edges $a_i b_j$ and $a_{i+1} b_j$ form an EC5 with $u$ and $v$.
}
    \label{fig:consecutive-EC5}
\end{figure}

\begin{proof}
If $e$ is a flip cut edge, then 
by Theorem~\ref{thm:Z-characterization}, $G_Z$ is disconnected, so we must have $|Z| \ge 2$. (For this direction we did not use the assumption about EC5's.) 

For the other direction, suppose $e$ is not a flip cut edge and $|Z| \ge 2$.  We will show that $e$ is a diagonal of an EC5.  By Theorem~\ref{thm:Z-characterization}, $G_Z$ is connected.  
Because neighbours in $G_Z$ are consecutive in the orderings of $A$ and $B$ (by Observation~\ref{obs:consecutive-neighbours}),
there must be two edges $f,g$ in $Z$ that are incident at one end and consecutive in the ordering of $A$ (or $B$) at the other end.  

Without loss of generality, suppose $f = a_{i}b_j$ and 
$g = a_{i+1}b_j$. 
See Figure~\ref{fig:consecutive-EC5}.
The five points $a_i, a_{i+1}, b_j, u, v$ form a convex pentagon. (Observe that the angle at $u$ is convex because of the EC4 $Q(f)$, and the angle at $a_i$ is convex because $a_i$ and $a_{i+1}$ make empty triangles with $uv$ and $a_i$ precedes $a_{i+1}$ in $A$, and similar arguments show that the other angles are convex.)  If this pentagon is not empty, then it contains a point $p$ outside $Q(f) \cup Q(g)$ and such a point of minimum $y$ coordinate is a point of $A$ between $a_i$ and $a_{i+1}$, contradition.  Thus the pentagon is empty, so $e$ is a diagonal of an EC5.  
\end{proof}

\subsubsection{Grid Point Sets}
\label{sec:grids}

The points of a $k \times \ell$ grid have no empty convex pentagons, so 
by Lemma~\ref{lemma:two-diagonals}, an edge is a flip cut edge if and only if $|Z| \ge 2$.
Edges ``near'' the boundary of the grid may fail to be flip cut edges, but we show that edges farther from the boundary are flip cut edges, and we show that in an infinite grid, all edges are flip cut edges.
We first deal with horizontal/vertical edges.

\begin{claim}
Let $e$ be a horizontal/vertical edge of a grid point set.  Then $e$ is a flip cut edge if and only if it does not lie on the boundary of the grid.
\end{claim}
\begin{proof}
Note that $e$ must have unit length since no edge goes through intermediate points.
We may assume without loss of generality that $e$ goes from $(x,y)$ to $(x+1,y)$.  
If $e$ lies on the boundary then $Z$ is empty.  
For the converse, suppose $e$ is not on the boundary.
Then the points just above $e$, $a_1 = (x,y+1), a_2=(x+1,y+1)$, and the points just below $e$, $b_1=(x+1,y-1),b_2=(x,y-1)$, lie in the grid, and the two edges
$a_1,b_1$ and $a_2,b_2$ lie in $Z$, so $e$ is a flip cut edge by Lemma~\ref{lemma:two-diagonals}.   
\end{proof}

Next consider an edge $e$ that is not horizontal or vertical. Reflect so that $e$ goes from the point $(x,y)$ to the point  $(x+\Delta x,y + \Delta y)$ with $\Delta x \ge \Delta y > 1$. 
Since an edge cannot go through intermediate points,
$\gcd(\Delta x,\Delta y) = 1$.
Let $L$ be the line through $e$. On the infinite grid,
translate $L$ upward and parallel to itself until it hits grid points and call the resulting line $L_U$.   Similarly, translate $L$ downward and parallel to itself until it hits grid points, and call the resulting line $L_D$.

\remove{Note that any point on $L_U$ makes an empty triangle with $e$, and the same for any point on $L_D$.  Thus, a grid point $a \in L_U$ and a grid point $b \in L_D$ provide an edge $ab \in Z$ iff $ab$ crosses $e$. In fact, we can prove that any edge of $Z$ must have endpoints of $L_U$ and $L_D$---we use that result in Figure~\ref{fig:grid-flip-graph_b} to demonstrate that some edges are not flip cut edges, but we do not use that result in our proofs.
\anna{Well, we could include the little proof . . . }
}
\begin{claim}
\label{claim:LU-LD}
A grid point makes an empty triangle with $e$ iff it lies on $L_U$ or $L_D$.
\end{claim}
\begin{proof}
The ``if'' direction is clear.  For the other direction, consider a triangle $T$ determined by $e$ and an apex grid point strictly above $L_U$.
The area of $T$ is strictly larger than the area of a triangle with apex on $L_U$.  By Pick's theorem (the area of a triangle on the grid is the number of grid points strictly inside the triangle plus half the number on the boundary of the triangle), triangle $T$ cannot be empty. 
\end{proof}

By the easy direction of this claim, a grid point $a \in L_U$ and a grid point $b \in L_D$ provide an edge $ab \in Z$ iff $ab$ crosses $e$. The other direction of the claim is used in Figure~\ref{fig:grid-flip-graph_b} to demonstrate that some edges are not flip cut edges.

We now analyze the points on $L_U$ and $L_D$.
We claim that there is a point $a = (a_x,a_y)$ on $L_U$ 
with 
$a_x \in [x,x+\Delta x)$
and 
$a_y \in (y,y+\Delta y]$.
To justify this, note that
if $a = (a_x,a_y)$ is a point on $L_U$, then so is $(a_x+\Delta x,a_y+\Delta y)$ and $(a_x- \Delta x,a_y- \Delta y)$.  This implies that there is a point $a = (a_x,a_y)$ on $L_U$ with 
$a_x \in [x,x+\Delta x)$.
And then we must have 
$a_y \in (y,y+\Delta y]$
otherwise $(a_x,a_y -1)$ would be between $L$ and $L_U$.    
By symmetry, the point $b = (b_x,b_y) := 
((x + \Delta x) - (a_x - x),
(y + \Delta y) - (a_y - y)) = (2x + \Delta x - a_x,2y + \Delta y - a_y)$
lies on $L_D$ and has
$b_x \in (x,x+\Delta x]$
and 
$b_y \in [y,y+\Delta y)$.
Then the segment $ab$ crosses $e$  at their midpoints, so $ab \in Z$.  
The same is true for the segment from $(a_x + i \Delta x,a_y + i \Delta y)$ to $(b_x - i \Delta x, b_y - i \Delta y)$ for any $i = 0, 1, \ldots$. 
This implies that on an infinite grid, $Z$ has infinite size.  Thus we have proved:

\begin{lemma}
\label{lem:infinite-grid}
For the points of the infinite integer grid, every edge is a flip cut edge.    
\end{lemma}

In a finite grid,  
the boundary edges are not flip cut edges, but they are not the only exceptions, see Figure~\ref{fig:grid-flip-graph}.  It is possible to characterize flip cut edge in terms of integer solutions to equations, but for now we simply note that
a large enough grid of $n$ points has $\Theta(n^2)$ flip cut edges, as justified by the following.

\begin{proposition} 
Let $G$ be 
a $k \times \ell$ grid and let $G'$ be the middle one-third. Then every edge of $G'$ is a flip cut edge of $G$.
\end{proposition}
\begin{proof}
Let $e$ be an edge of $G'$. We will show that $Z$ has at least two elements. 
As above, we may assume that $e$ goes from $(x,y)$ to $(x + \Delta x,y+ \Delta y)$ with $\Delta x \ge \Delta y >1$ and $\gcd(\Delta x,\Delta y) =1$. Since $e$ lies in $G'$ we have $\Delta x \le k/3$, $\Delta y \le \ell/3$.

As noted above, there is a point 
$a = (a_x,a_y)$ on $L_U$ with
$a_x \in [x,x+\Delta x)$
and 
$a_y \in (y,y+\Delta y]$
and a point 
$b = (b_x,b_y) = (2x + \Delta x - a_x,2y + \Delta y - a_y)$
on $L_D$ with  
$b_x \in (x,x+\Delta x]$
and 
$b_y \in [y,y+\Delta y)$.
Both $a$ and $b$ lie in the grid $G$ and $ab \in Z$.
For our second element of $Z$, we use the edge from $(a_x + \Delta x,a_y + \Delta y)$ on $L_U$ to $(b_x - \Delta x, b_y - \Delta y)$ on $L_D$. Note that these points lie in $G$.
\end{proof}

\begin{figure}[ht]
    \begin{subfigure}[t]{0.45\textwidth}
        \centering
          \includegraphics[width=\textwidth]{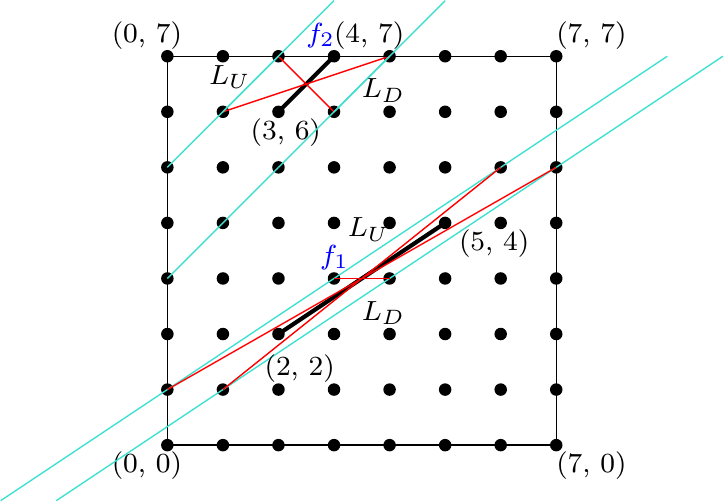}
          \caption{$f_1$ and $f_2$ are flip cut edges since they are in at least two EC4's, as shown by the diagonals (in red) with endpoints on $L_U$ and $L_D$ (in cyan). 
}
          \label{fig:grid-flip-graph_a}
    \end{subfigure}%
    \hfill
    \begin{subfigure}[t]{0.45\textwidth}
          \centering
          \includegraphics[width=\textwidth]{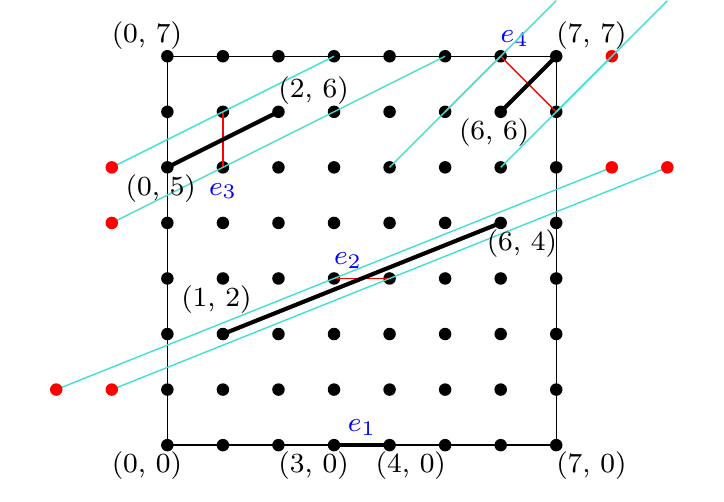}
          \caption{$e_1, e_2, e_3, e_4$ are not flip cut edges since they are in 0 or 1 EC4's.
}
          \label{fig:grid-flip-graph_b}
    \end{subfigure}
    \caption{Example edges in a $7 \times 7$ grid.
    }
    \label{fig:grid-flip-graph}
\end{figure}

\label{sec:channels}
\subsubsection{The Number of Flip Cut Edges}

In Section~\ref{sec:examples}, we claimed that channels, as shown in Figure~\ref{fig:channel}, are point sets with $\Theta(n^2)$ flip cut edges.  Here we justify that claim.

\begin{proposition}
For a channel, every edge between an interior point on the upper reflex chain and an interior point on the lower reflex chain
is a flip cut edge.
\end{proposition}
\begin{proof}
Let the points of the upper reflex chain be $t_1, \ldots, t_n$ and the points of the lower reflex chain be $b_1, \ldots, b_n$, as shown in Figure~\ref{fig:channel}.
Consider the edge $b_i t_j$ where $i,j \notin \{1,n\}$.  
We use Lemma~\ref{lemma:two-diagonals} to show that $b_i t_j$ is a flip cut edge.
There is no  EC5 with $b_i t_j$ as a diagonal because such an EC5 would have to have at least 3 points of one reflex chain, say the upper one, and those points cannot be part of a convex polygon with $b_i$. 
However, both the edges $b_{i-1} t_{j+1}$ and $b_{i+1} t_{j-1}$ form EC4's together with $b_i t_j$.
Thus by Lemma~\ref{lemma:two-diagonals}, $b_i t_j$ is a flip cut edge. 
\end{proof}

\subsubsection{The Number of Components from a Flip Cut Edge}
\label{sec:hourglasses}

In Section~\ref{sec:examples}, we claimed that one flip cut edge $e$ can cause $\Theta(n)$ disconnected components in ${\cal F}_{-e}$.   
Here we justify that claim.

\begin{theorem} 
A flip cut edge can create $O(n)$ disconnected components in the flip graph, and this is the most possible.
\end{theorem}

\begin{proof}
First of all, connected components in the flip graph ${\cal F}_{-e}$ correspond to connected components in $G_Y$ (by Theorem~\ref{thm:new-characterization-2}) and therefore correspond to disjoint sets of points, so there are at most $n$ of them.

The hourglass in Figure~\ref{fig:dis_comp} is 
an example of a point set of size $2n+2$ with a flip cut edge $e$ that results in $n$ components in ${\cal F}_{-e}$.  
To construct the hourglass,  place points $a_1, \ldots, a_n$ along the top half of a circle, and let $b_1, \ldots, b_n$ be the diametrically opposite points.  For each $i$ construct the wedge from $a_i$ to $b_{i-1}$ and $b_{i+1}$.  Place points $u$ and $v$ on the horizontal diameter of the circle, so close to the center of the circle that they are \emph{inside} all the wedges.  Then $G_Y$ consists of the $n$ edges $a_i b_i$, and no two of these are connected in $G_Y$.  Thus (by Theorem~\ref{thm:new-characterization-2}) ${\cal F}_{-e}$ has $n$ disconnected components. 
\end{proof}

\section{Flip Cut Number for Points in Convex Position}



As mentioned in the Introduction, a point set in convex position has no flip cut edge, i.e., its flip cut number is greater than 1.  
In this section we show that the flip cut number of $n$ points in convex position is $n-3$. 

We give a direct proof, but first 
we discuss what the result means in terms of associahedra.
The flip graph of $n$ points in convex position is the 1-skeleton of the $(n-3)$-dimensional associahedron ${\cal A}_{n-1}$
(see Figure \ref{fig:A_5}) so by Balinski's theorem~\cite{balinski1961graph}, the 1-skeleton is $(n-3)$-connected. 
The dual 
polytope $\bar{\cal A}_{n-1}$ is also an $(n-3)$-dimensional polytope with an $(n-3)$-connected 1-skeleton. 
The face lattice of $\bar{\cal A}_{n-1}$ is the inverted face lattice of ${\cal A}_{n-1}$, so the vertices of $\bar{\cal A}_{n-1}$ correspond to the $(n-4)$-dimensional faces (the facets) of ${\cal A}_{n-1}$ and the  edges of $\bar{\cal A}_{n-1}$ correspond to $(n-5)$-dimensional faces of  ${\cal A}_{n-1}$. 
Forbidding a chord of the original polygon corresponds to deleting a facet of ${\cal A}_{n-1}$, i.e., a vertex of $\bar{\cal A}_{n-1}$.
Deleting fewer than $n-3$ facets of ${\cal A}_{n-1}$ leaves the remaining facets connected via $(n-5)$-dimensional faces. 
What our result shows is that 
the remaining vertices (0-dimensional faces) of ${\cal A}_{n-1}$ are connected via the remaining edges (1-dimensional faces) of ${\cal A}_{n-1}$.
We do not see how to prove our result using this polyhedral interpretation.   
Instead, we give a direct proof that the flip cut number of $n$ points in convex position is $n-3$.

For points in convex position, we number the points $p_1, \ldots, p_n$ in cyclic order around the convex hull.  The edges $p_i p_{i+1}$ (addition modulo $n$) must be present in every triangulation.  A \defn{chord} is a line segment joining points that are not consecutive in the cyclic ordering. An \defn{ear} 
is a 
chord of the form $p_{i-1} p_{i+1}$ (addition modulo $n$), and we say that this ear \defn{cuts off} point $p_{i}$.
Every triangulation of a convex point set on $n \ge 4$ points has at least two ears (because the dual tree has at least two leaves).
A \defn{star triangulation} is a triangulation all of whose chords are incident to the same point.

For a set $X$ of chords, we study connectivity of ${\cal F}_{-X}$. 
We define the \defn{degree} of a point $p$ in $X$ to be the number of chords of $X$ incident to $p$. 

We begin with a necessary condition for a flip cut set. This provides a rigourous proof that the flip cut number is greater than 1, and will also be needed for our main result. 

\begin{lemma}
\label{lemma:F-touches-all}
If $X$ is a flip cut set for a set of points in convex position, then 
every point is incident to an edge of $X$, i.e., the degree in $X$ of every point is at least 1.
 \end{lemma}

\begin{proof}
We prove the contrapositive: if there is a point $p \in P$ with no incident edges of $X$, then $X$ is not a flip cut set.
Specifically, we prove that any triangulation $T$ in ${\cal T}_{-X}$ is connected to the star triangulation centered at $p$.
This is in fact the standard way to show connectivity of the (full) flip graph of a convex point set.  In our situation, we must show that the required flips do not use forbidden edges. 

So, suppose triangulation $T$ contains a triangle 
incident to $p$ whose other two points, $q$ and $r$, are not consecutive on the convex hull. 
Now $T$ contains a triangle, say $qrs$, on the other side of $qr$. Then $pqsr$ forms a convex quadrilateral, and we can flip the chord $qr$ to the chord $ps$ (which is not forbidden).  This increases the number of chords incident to $p$, so, by induction, $T$ can be flipped to the star triangulation centered at $p$ without using forbidden edges. 
\end{proof}

\begin{corollary}
    \label{cor:star}

There is no flip cut edge in a convex point set.
\end{corollary}

In the rest of this section we prove that the flip cut number for $n$ points in convex position is $n-3$.  
We begin with the following two lemmas.

\begin{lemma}
    \label{lemma:flip cut number}
A convex point set of size $n$ has a flip cut set of size $n-3$.
\end{lemma}

\begin{figure}[t]
    \begin{subfigure}[t]{0.45\textwidth}
        \centering
          \includegraphics[width=0.6\textwidth]{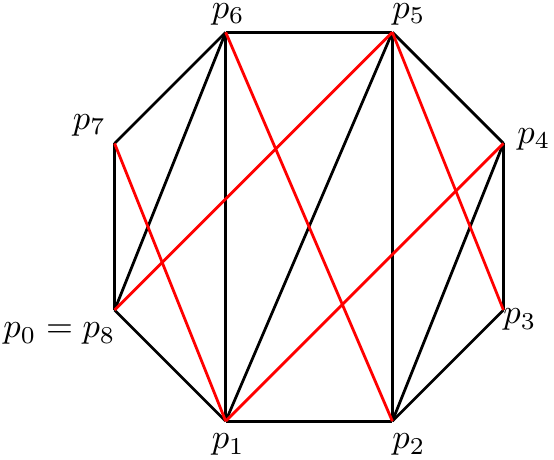}
          \caption{A zigzag triangulation $T$ (in black) and 
          a set of forbidden edges $X$ (in red) that makes $T$ a frozen triangulation.
          } 
          \label{fig:frozen-zigzag}
    \end{subfigure}%
    \hfill
    \begin{subfigure}[t]{0.45\textwidth}
          \centering
          \includegraphics[width=0.6\textwidth]{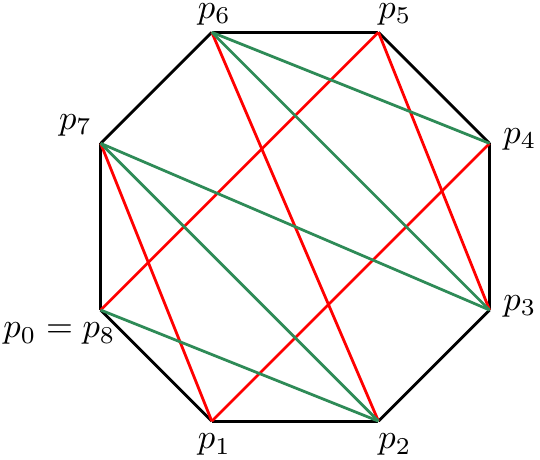}
          \caption{Another zigzag triangulation $T'$ in ${\cal F}_{-X}$.
          }
          \label{fig:other-frozen-zigzag}
    \end{subfigure}
    \caption{Construction of a flip cut set of size $n-3$
    }
    \label{fig:n-3-flip-cut-set}
\end{figure}

\begin{proof}

Consider a ``zigzag'' triangulation $T$ of a convex $n$-gon as shown in Figure~\ref{fig:frozen-zigzag}. 
This triangulation has $n-3$ chords that form a
path $p_0,p_{n-2},p_1,p_{n-3},..., p_{\lfloor \frac{n}{2} \rfloor}$.   
Each chord $e$ of $T$ can flip to a unique other chord $f(e)$.  If we forbid the $n-3$ chords $X := \{f(e) : e \text{ a chord of } T \}$ then $T$ becomes a frozen triangulation, i.e., it becomes an isolated node in the flip graph ${\cal F}_{-X}$. In order to complete the proof, we only need to show that there exists another 
triangulation 
in ${\cal F}_{-X}$
different from $T$. 
(This is where we use that $T$ is a zig-zag triangulation---for example, a star triangulation would not work.)
We create another zigzag triangulation $T'$ in the following way. 
The path of chords for $T'$ starts at the same vertex $p_0=p_n$  but the first chord from $p_0$ goes in the other direction, i.e., to $p_2$.
See Figure \ref{fig:other-frozen-zigzag}. The new zigzag path 
is
$p_0,p_2,p_{n-1},p_3,p_{n-2},...p_{\lceil \frac{n}{2} \rceil}$. We prove that $T'$ does not use any forbidden edges. Note that each forbidden edge crosses only one chord of $T$. However, each chord of $T'$, crosses at least two chords of $T$. This is because each chord in $T'$ of the form $p_{i+2}p_{n-i}$ for $i=0,..,\lfloor \frac{n}{2} \rfloor-2$ crosses $p_{i+1}p_{n-i-2}$ and $p_{i+1}p_{n-i-3}$ from $T$ and any chord of the form $p_{i+1}p_{n-i}$ for $i=1,..,\lfloor \frac{n}{2} \rfloor-2$ in $T'$ crosses $p_{i}p_{n-i-1}$ and $p_{i-1}p_{n-i-1}$ from $T$. So, since each forbidden edge $e \in X$ crosses exactly one chord of $T$, but each chord of $T'$ crosses at least two chords of $T$, therefore $T'$ does not contain any forbidden edges.
\end{proof}

\begin{lemma}
    \label{lemma:existence_of_triangulation}
    Consider a convex point set $P$ with $n$ points
    and a set $X$ of forbidden edges with $|X| \le n-3$.
    Then there is a triangulation of $P$ that uses no forbidden edges, i.e., ${\cal T}_{-X}$ is non-empty.
\end{lemma}

\begin{proof}
The proof is by induction on $n$ with base case $n=3$.
If there is a point $p$ that is incident to no forbidden edges, then take the star triangulation at $p$.
Otherwise, suppose every point is incident to at least one forbidden edge. There are $n$ chords of the form $p_{i-1} p_{i+1}$ (addition modulo $n$).  
So, there exists at least one of them that is not forbidden. Suppose it is $p_{i-1} p_{i+1}$. 
We build a triangulation containing the ear $p_{i-1} p_{i+1}$. 
Let $P'$ be $P - \{p_i\}$ of size $n' = n-1$.  Since $p_i$ is incident to at least one forbidden edge, there are at most $n-4 = n' -3$ forbidden edges of $P'$.  Thus, by induction, there is a triangulation of $P'$ that uses no forbidden edges.  Together with chord $p_{i-1} p_{i+1}$ this provides the desired triangulation of $P$.  Thus ${\cal T}_{-X}$ is non-empty.
\end{proof}

\begin{theorem}
The flip cut number of a convex point set 
$P$ with $n$ points
is $n-3$.
\end{theorem}

\begin{proof}
By Lemma \ref{lemma:flip cut number}, we can disconnect $P$'s flip graph by forbidding $n-3$ edges. So, now we only need to show that 
if $X$ is a set of forbidden edges 
and $|X| \le n-4$,
then ${\cal F}_{-X} (P)$ is connected.
By Lemma~\ref{lemma:F-touches-all}, ${\cal F}_{-X} (P)$ is connected if there is a point $p \in P$ not incident to any edge of $X$.  Thus we may assume that every point $p \in P$ is incident to some edge of $X$. 
Let $S$ and $T$ be two triangulations of $P$ that do not contain any edges of $X$.  We will prove that $S$ is connected to $T$ in ${\cal F}_{-X} (P)$. 
We will prove this by induction on $n$, with the base case $n=3$ where the statement is vacuously true.

We know that each triangulation contains at least two ears.
We consider two cases. 

\noindent
\textbf{Case 1.} 
$S$ and $T$ have a shared ear 
that cuts off point $p_i$. 
Consider sub-polygon $P'$ obtained by removing point $p_i$.
By assumption
there is a forbidden edge incident to $p_i$.
Let $X'$ be the forbidden edges of $X$ not incident to $p_i$.
Then $|X'| \le |X|-1 \le |P|-5 = |P'| - 4$ so
by induction, the triangulations of $P'$ induced from $S$ and $T$ 
are connected in ${\cal F}_{-X'}(P')$.
Thus $S$ and $T$ are connected in ${\cal F}_{-X} (P)$.

\noindent
\textbf{Case 2.} 
$S$ and $T$ have no shared ear.  We 
claim that there is an ear $e_1$ in $S$ and an ear $e_2$ in $T$ such that $e_1$ and $e_2$ do not cross, i.e., the points that are cut off by $e_1$ and by $e_2$ are not adjacent on the convex hull. 
Suppose $S$ has an ear $e_1$ that cuts off point $p_i$.  $T$ has at least two ears, and if they cross $e_1$, they must cut off points $p_{i-1}$ and $p_{i+1}$.  But then the second ear of $S$ cannot cross both these ears of $T$ unless $n=4$ (and $X= \emptyset$), in which case, a single flip converts $S$ to $T$. 

Thus we may assume
a pair of non-crossing ears $e_1$ of $S$ and $e_2$ of $T$.
Suppose $e_1$ cuts off point $q_1$ and $e_2$ cuts off point $q_2$.
Let $P_1$ be $P - \{q_1\}$.  $P_1$ has $n_1 = n-1$ points.
Let $X_1$ be the forbidden edges of $X$ induced on $P_1$. 
By assumption
there is at least one edge of $X$ incident to $q_1$, so $|X_1| \le |X| -1 \le n-5 = n_1 - 4$.  By induction, the flip graph $\mathcal{F}_{-X_1}(P_1)$
is connected.
Let $S_1$ be the triangulation of $P_1$ formed by cutting the ear $e_1$ off $S$.
The plan is to apply induction on $P_1$ to connect triangulation $S_1$ to a new triangulation $R_1$ of $P_1$ that includes the chord $e_2$. 

We construct $R_1$ as follows. Let $P_{2}$ be $P - \{q_1,q_2\}$.
Then $P_2$ has size $n_2 = n-2$. 
Let $X_{2}$ be the forbidden edges of $X$ induced on $P_{2}$.  Then
$|X_{2}| \le |X_1| \le n_1 - 4 = n_{2} - 3$. 
By Lemma~\ref{lemma:existence_of_triangulation}, $P_{2}$ has a triangulation that uses no chords of $X_{2}$. Adding chord $e_2$ yields the triangulation $R_1$. 

By induction on $P_1$ and $X_1$ the triangulations $S_1$ and $R_1$ are connected in $\mathcal{F}_{-X_1}(P_1)$.  
Finally, $R_1$ and $T$ share the ear $e_2$, so by Case 1, they are connected in 
$\mathcal{F}_{-X}(P)$.
Altogether, we have connected $S$ to $T$ in $\mathcal{F}_{-X}(P)$, as required.
\end{proof}

\section{Conclusions and Open Problems}

We examined connectivity of the flip graph of triangulations when some edges between points are forbidden, and introduced the concepts of flip cut edges, flip cut sets, and the flip cut number. 
We gave an $O(n \log n)$ time algorithm to identify flip cut edges and test connectivity after forbidding a flip cut edge, and we proved that the flip cut number of a convex $n$-gon is $n-3$.
We conclude with some open questions.
\begin{enumerate}
\item Is there a polynomial time algorithm to
test if a set of edges is a flip cut set?  To compute the flip cut number?

\item 
The asymptotic diameter of the flip graph of a convex $n$-gon is $2n-10$---a famous result of Sleater, Tarjan and Turston~\cite{sleator1988rotation}, improved to all $n > 12$ by Pournin~\cite{pournin2014diameter}.
For a convex point set and a set $X$ of forbidden edges with $|X| < n-3$ what is the diameter of ${\cal F}_{-X}$ in terms of $n$ and $|X|$?

\item It is open whether there is a polynomial time algorithm for the flip distance problem for a convex polygon.
Is the following generalization 
NP-complete: Given $n$ points in convex position, a set $X$ of forbidden chords, two triangulations $T_1$ and $T_2$, and a number $k$, is the flip distance from $T_1$ to $T_2$ in ${\cal F}_{-X}$ less than or equal to $k$?

\item A flip cut edge is a bottleneck to connecting triangulations via flips.
One might guess that point sets with no flip cut edges provide better mixing properties. More generally, how does the flip cut number affect mixing properties?

\end{enumerate}

\section*{Acknowledgements}
For helpful comments we thank Lionel Pournin and participants of the 2022 Banff International Research Station (BIRS) workshop on Combinatorial Reconfiguration.

\bibliography{main}

\end{document}